\documentclass[a4paper,USenglish,cleveref]{lipics-v2021}

\usepackage{amsmath}
\usepackage{amssymb}
\usepackage{amsthm}
\usepackage{booktabs}
\usepackage{array}
\usepackage{xcolor}
\usepackage{subcaption}
\usepackage{prettyref}
\usepackage{wrapfig}
\usepackage{url}
\usepackage{hyperref}
\usepackage{tikz}
\usetikzlibrary{calc,arrows.meta}
\title{ParlayMarket: Automated Market Making for Parlay-Style Joint Contracts}
\titlerunning{ParlayMarket}

\author{Ranvir Rana}
  {Kaleidoscope Blockchain Inc., USA}
  {}
  {}
  {}

\author{Viraj Nadkarni}
  {Princeton University, Princeton, NJ, USA}
  {}
  {}
  {}

\author{Niusha Moshrefi}
  {Princeton University, Princeton, NJ, USA}
  {}
  {}
  {}

\author{Pramod Viswanath}
  {Princeton University, Princeton, NJ, USA}
  {}
  {}
  {}

\authorrunning{R. Rana, V. Nadkarni, N. Moshrefi, and P. Viswanath}
\Copyright{Ranvir Rana, Viraj Nadkarni, Niusha Moshrefi, and Pramod Viswanath}



\nolinenumbers

\EventEditors{Aggelos Kiayias and Maria Kyropoulou}
\EventNoEds{2}
\EventLongTitle{8th Conference on Advances in Financial Technologies (AFT 2026)}
\EventShortTitle{AFT 2026}
\EventAcronym{AFT}
\EventYear{2026}
\EventDate{October 6--9, 2026}
\EventLocation{London, UK}
\EventLogo{}
\SeriesVolume{395}
\ArticleNo{19}

\begin{document}
\maketitle


\begin{abstract}
Prediction markets are powerful mechanisms for information aggregation, but existing designs are optimized for single-event contracts. Traders frequently express beliefs about joint outcomes - sports parlays, conditional forecasts, multi-scenario financial bets. Current platforms either prohibit such trades or rely on ad hoc mechanisms that ignore correlation structure, resulting in inefficient prices and fragmented liquidity. We introduce \textbf{\textit{ParlayMarket}}, an automated market-maker for parlay-style joint contracts. The mechanism maintains a shared pairwise exponential-family belief state, so all base and parlay prices are marginals of one coherent distribution. This compresses the $2^M$ outcome space into $O(M^2)$ sufficient statistics and allows one liquidity pool to support an exponentially large family. 
Our main result characterizes the resulting learning and loss dynamics. Under repeated
trading, prices converge to the best pairwise approximation of the true joint distribution.
The induced expected market-maker loss grows at most quadratically in the number of base events,
rather than exponentially in the number of listed parlays; moreover, this quadratic dependence
is worst-case optimal for dense pairwise dependence, since there are $O(M^2)$ independent
correlation directions to learn. Parlay trades are essential to this guarantee: they provide
direct constraints on joint outcomes and reduce steady-state error relative to learning from
marginal trades alone. Experiments on synthetic correlated markets and historical Kalshi
combo data confirm the predicted scaling and show that the mechanism remains effective in
realistic market-making settings. Our results demonstrate that combinatorial expressiveness does not require combinatorial capital.

\end{abstract}


\section{Introduction}
\label{sec:intro}

\subparagraph{Two uses of knowledge in a society} Knowledge in an economy can be used in two ways---\emph{factually} and \emph{counterfactually}. Factual knowledge concerns statements such as ``the price of oil will be greater than \$100 on March 29'' or ``the election will be won by party $P$.'' Counterfactual knowledge, by contrast, concerns statements such as ``if the price of oil is greater than \$100 on March 29, then the price of a banana will be greater than \$2 on April 30'' or ``if the election is won by party $P$, then the probability of a recession will increase.'' Prediction markets have proven to be a powerful mechanism for aggregating and eliciting factual knowledge. The prices quoted by a prediction market can serve as a public good, allowing uninformed members of the economy to peer into the future through a market-implied best guess and plan accordingly.

\subparagraph{The liquidity problem} Despite this promise, prediction-market liquidity is highly uneven. Existing markets tend to have substantial liquidity in sensational or gambling-oriented domains such as sports, while economically useful markets, including public policy, local politics, and macroeconomic indicators, are often woefully illiquid \cite{freesystems2024truth}. Automated market makers (AMMs) are particularly well suited to such illiquid markets because they can provide continuous liquidity without requiring a counterparty \cite{hanson2002lmsr,hanson2003comb}. However, conventional AMMs incur losses on their locked liquidity in exchange for eliciting information. In this sense, subsidies are required to induce traders to reveal factual knowledge.

\subparagraph{The joint events problem} A further challenge arises when the information of interest is not merely factual, but counterfactual. Prediction markets aggregate dispersed information into prices, but practical automated market makers are primarily designed for single-event contracts. In particular, bounded-loss pricing for single markets is typically achieved through Hanson's logarithmic market scoring rule (LMSR) \cite{hanson2002lmsr,hanson2003comb}. In many applications, however, traders reason about joint outcomes: sports parlays, conditional forecasts across related events, or multi-scenario financial bets. The probability of such a contract depends not only on the marginal probabilities of its constituent events, but also on the dependence structure among them.

\subparagraph{Computational complexity} Supporting such joint contracts introduces a fundamental challenge. With $M$ binary events, the joint outcome space contains $2^M$ possible realizations. A naive combinatorial market maker must maintain prices and risk exposure across all of these outcomes, leading to worst-case loss and computational requirements that scale exponentially in $M$. This renders direct approaches intractable beyond very small systems. The challenge is well known in combinatorial prediction markets: allowing direct trading over joint outcome spaces provides expressive contracts, but typically incurs exponential state and inference costs \cite{hanson2003comb,chen2008bn}. Existing practical systems therefore either avoid parlay-style contracts or handle them through ad hoc mechanisms such as request-for-quote workflows, which fragment liquidity and need not maintain a coherent joint distribution.

\subparagraph{Our contributions} This paper shows that \emph{combinatorial expressiveness does not require combinatorial capital}. We introduce \textit{ParlayMarket}, an automated market maker that maintains a shared probabilistic belief over all outcomes and derives every base and parlay price as a marginal of this belief. The key design choice is to restrict the belief state to a pairwise exponential family. Rather than tracking one state variable per joint outcome, the mechanism maintains only
\[
M+\binom{M}{2}=O(M^2)
\]
sufficient statistics: one marginal parameter per event and one interaction parameter per pair. This is the smallest generic exponential-family representation capable of encoding arbitrary pairwise dependence: independence cannot represent correlations, while an unrestricted joint model is exponential. Statistically, this model is also the maximum-entropy distribution matching the learned first and second moments, and therefore does not introduce higher-order structure unsupported by the observed trade flow \cite{wainwright2008graphical}.

Our approach is related in spirit to inference-based views of prediction markets, where prices are treated as marginals of an underlying probabilistic model \cite{frongillo2012stochastic,chen2008bn,laskey2012graphical}. The key difference is that ParlayMarket learns the dependence structure online from endogenous order flow rather than imposing a fixed graphical structure ex ante. Trades in both base markets and parlay markets are interpreted as observations about the underlying joint distribution. After each trade, the mechanism updates its shared parameters, allowing information from one contract to propagate immediately to related contracts. This couples liquidity and learning across the entire contract family rather than maintaining separate books for each parlay.

We show that this design satisfies three central objectives. First, it pools liquidity across markets by pricing all contracts from a shared state, avoiding fragmentation. Second, it achieves bounded loss through an LMSR-style cost function defined over this shared representation, with loss that grows at most quadratically in $M$, rather than exponentially in the number of joint outcomes. Moreover, this quadratic dependence is worst-case optimal for dense pairwise dependence. Third, it aggregates information efficiently: as trades accumulate, the model converges to the best pairwise approximation of the true joint distribution under the observed data. This does not claim that higher-order dependence is absent; rather, it identifies the pairwise family as the practical frontier at which global coherence, online learnability, and non-combinatorial loss can all be achieved simultaneously.

Conceptually, ParlayMarket transforms the market into an online inference system. Empirically, controlled simulations confirm the predicted polynomial-versus-exponential loss separation, as shown in Figure~\ref{fig:main}, and a replay on historical Kalshi combo data shows that the mechanism remains effective as a practical correlation-aware quoting engine. Together, these results show that parlay-style prediction markets can be supported by an automated market maker
without full joint enumeration or separate capitalization for every combination.

\begin{figure}[t]
  \centering
  \begin{subfigure}[t]{0.49\columnwidth}
    \includegraphics[width=\columnwidth]{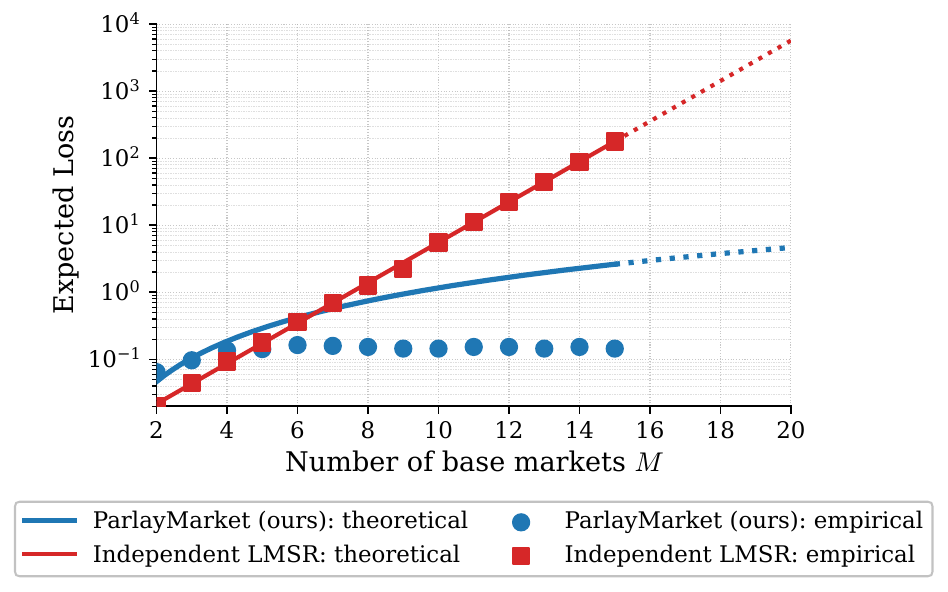}
    \caption{Expected Loss}
    \label{fig:main-mean}
  \end{subfigure}
  \hfill
  \begin{subfigure}[t]{0.49\columnwidth}
    \includegraphics[width=\columnwidth]{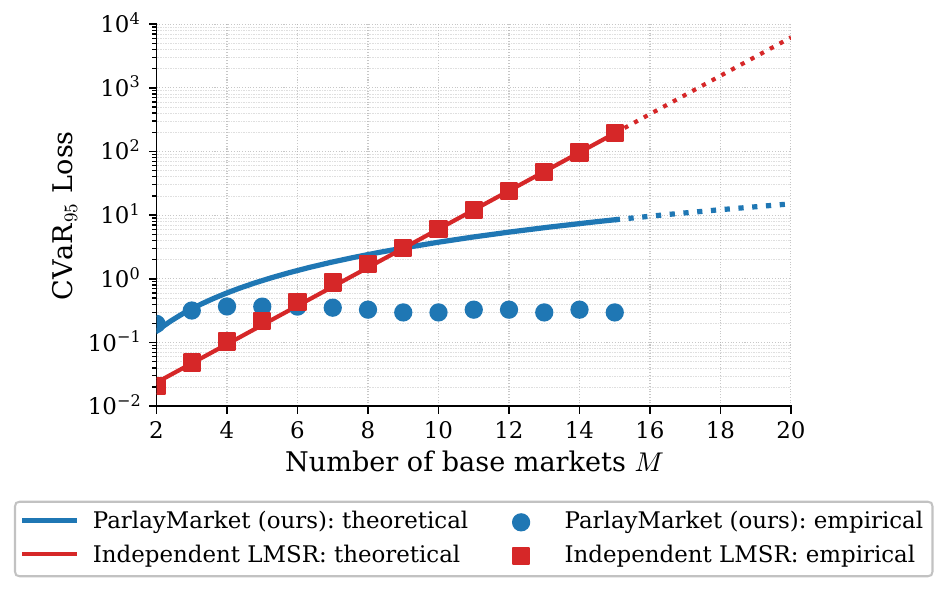}
    \caption{$\mathrm{CVaR}_{95}$ Loss}
    \label{fig:main-max}
  \end{subfigure}
  \caption{ParlayMarket enables $2^M$ markets with $O(M^2)$ capital. Theoretical curves show the asymptotic loss scaling implied by the analysis in Section \ref{sec:correlation-mm}. Empirical curves are obtained from controlled simulation of correlated binary markets under the Gaussian score model described later in Sections \ref{sec:objectives} and \ref{sec:evaluation}, and confirm the same polynomial-vs.-exponential separation in practice.}
  \label{fig:main}
\end{figure}

\subparagraph{Paper Overview.}
The remainder of the paper is organized as follows. In Section \ref{sec:background}, we review background on prediction markets, combinatorial and parlay markets, and probabilistic models of dependence, and position our work relative to prior approaches. In Section~\ref{sec:objectives}, we formalize the problem and define the design objectives for parlay-capable market makers. Section~\ref{sec:baseline} analyzes existing approaches and shows that they fail to simultaneously satisfy these objectives due to fragmented liquidity and combinatorial loss. Section~\ref{sec:correlation-mm} introduces ParlayMarket and describes its pricing and update mechanisms. We then present the theoretical results on convergence, loss bounds, and information aggregation. Finally, Section~\ref{sec:evaluation} evaluates the mechanism empirically, and Section~\ref{sec:conclusion} discusses implications and future directions.

\section{Background and Related Work} 
\label{sec:background}

\subparagraph{Cost-function market makers}
Prediction markets are commonly implemented using cost-function market makers, most notably Hanson's logarithmic market scoring rule (LMSR)~\cite{hanson2003comb,hanson2007log}. LMSR provides continuous liquidity, proper-scoring-rule prices, and bounded worst-case loss. Chen and Pennock~\cite{chen2008complexity} interpret LMSR as a utility-based market maker, while Abernethy et al.~\cite{abernethy2013market} connect market scoring rules to convex optimization, online learning, and Bregman divergence minimization. Related work on exponential-family markets~\cite{abernethy2014expfam} shows how bounded-loss market makers can be understood as maintaining beliefs in a parametric family. CFMMs provide a parallel AMM tradition in decentralized finance~\cite{angeris2020cfmm,milionis2022lvr}, but these mechanisms do not by themselves enforce probabilistic consistency across related event
contracts.

\subparagraph{Combinatorial and Parlay Markets}
Combinatorial prediction markets allow trading over joint outcome spaces~\cite{hanson2003comb}, but exact implementations require reasoning over \(2^M\) outcomes and become intractable for large \(M\). Prior work therefore restricts the security space or market geometry to obtain tractable loss bounds~\cite{hossain2012geometric}, or studies settings where simple trader interfaces fail to recover coherent joint beliefs~\cite{othman2010simple}. In practice, platforms such as Polymarket and Kalshi~\cite{polymarket,kalshi} rarely implement full combinatorial AMMs; parlay-like products are typically handled through RFQ workflows or separate markets, which fragment liquidity and need not maintain a single coherent joint distribution.

\subparagraph{Graphical models and inference-based markets}
Graphical models offer a natural way to represent dependence compactly. Pairwise Markov random fields, including Ising models, represent binary joint distributions using only first-order and pairwise interaction terms and correspond to maximum-entropy distributions matching specified marginals and pairwise moments~\cite{wainwright2008graphical}. Learning such models has been studied extensively~\cite{ravikumar2010ising,bresler2015ising, vuffray2016screening}, and online variants connect naturally to stochastic optimization and online convex programming~\cite{zinkevich2003ogd,hazan2019oco}. Prediction markets have also been interpreted as probabilistic inference systems~\cite{frongillo2012stochastic, chen2010eliciting}; Bayesian-network and graphical-model market makers make this connection explicit~\cite{chen2008bn,laskey2012graphical}. Unlike those mechanisms, \textit{ParlayMarket} learns the dependence structure endogenously from singleton and parlay order flow while using one shared state to quote all contracts coherently.

\section{Model and Design Objectives}
\label{sec:objectives}

We formalize automated market making for correlated prediction markets as a joint mechanism design and statistical inference problem. The market maker interacts with a stochastic stream of trades that reveal partial information about an underlying joint distribution, and must simultaneously provide liquidity, control risk, and learn the dependence structure.

\subsection{Market Setting}

We consider a collection of $M$ binary events $X_1, \dots, X_M \in \{0,1\}$, and denote a joint outcome by $x \in \{0,1\}^M$. The uncertainty over outcomes is governed by an unknown ground-truth distribution $P^{*}$ over $\{0,1\}^M$, which captures both marginal probabilities and dependencies among events.

\subparagraph{Contracts.}
The market supports trading in a family of payoff functions indexed by subsets $S \subseteq \{1,\dots,M\}$. Each contract corresponds to a Boolean payoff

\begin{equation}
    f_S(x) = \prod_{i \in S} x_i
    \label{eq:payoff}
\end{equation}
which pays $1$ if all events in $S$ occur and $0$ otherwise. In particular, singleton sets $S=\{i\}$ recover base contracts, while larger subsets correspond to parlay contracts. Let $\mathcal{F}$ denote the set of all supported contracts.

\subparagraph{Pricing and coherence.}
At any time, the market maker maintains a parametric distribution $P_{\boldsymbol{\varphi}}$ over $\{0,1\}^M$, which induces prices for all contracts via expectations:

\begin{equation}
    p_S(\boldsymbol{\varphi}) = \mathbb{E}_{x \sim P_{\boldsymbol{\varphi}}}[f_S(x)]
    \label{eq:core_price}
\end{equation}
This ensures probabilistic coherence, in the sense that all prices are consistent with a single joint distribution. In particular, marginal and joint prices satisfy the usual consistency constraints implied by $P_{\boldsymbol{\varphi}}$.

\subsection{Trader/Information Model}
We model the market as an online information aggregation process in which traders arrive sequentially and trade based on private signals about the underlying outcome distribution.

\subparagraph{Underlying process}

The underlying uncertainty in each base market is driven by a correlated Gaussian score
process. Specifically, there are $M$ base markets where market $i$ resolves \textsc{Yes}
if an underlying score $S_i(T)$ exceeds threshold $K_i$ at expiry $T$. The scores follow
correlated arithmetic Brownian motion
\begin{equation}
      dS_i(t) = \sigma_i\, dW_i(t), \qquad \mathrm{Cov}(dW_i, dW_j) = \rho_{ij}\, dt,
\end{equation}
so that $S_T \mid S_t \sim \mathcal{N}(S_t, \Sigma)$ with $\Sigma_{ij} = \rho_{ij}\sigma_i\sigma_j(T-t)$.
The binary outcome of market $i$ is $X_i = \mathbf{1}[S_i(T) > K_i]$.

The Gaussian score model is a natural choice for several reasons. First, it provides a
tractable closed-form expression for the conditional probability of each outcome given
the current score, enabling informed traders to compute exact signals. Second, the
pairwise correlation structure $\rho_{ij}$ maps cleanly onto the interaction parameters
$W_{ij}$ of the Ising approximation used by ParlayMarket, making the theoretical analysis
transparent. Third, it is a standard model for continuous-time prediction markets and
sports scoring processes, and admits straightforward emulation from real-world data by
calibrating $(\sigma_i, \rho_{ij}, K_i)$ to observed outcome frequencies and co-occurrences \cite{robinson2024pmamm}. In particular, it captures how the volatility in the likelihood of an event increases as we move closer to when the outcome is going to be revealed.

\subparagraph{Trader types.}
At each time step, a trader arrives and interacts with a single contract. \textit{Informed traders} observe the true underlying score process
$S_t$ and target the exact conditional probability. For base market $i$ they submit
$\hat{p}_i = \Phi\!\left(\frac{S^i_t - K_i}{\sigma_i\sqrt{T-t}}\right)$, and for parlay
$\mathcal{S}$ they submit the exact multivariate normal orthant probability
$P\!\left(\bigcap_{i\in\mathcal{S}}\{S_i(T)>K_i\} \mid S_t\right)$. \textit{Noise traders} observe a noisy signal
$\tilde{S}^i_t = S^i_t + \varepsilon^i$ with
$\varepsilon \sim \mathcal{N}(0, \mathrm{diag}(\sigma^2_i(T-t)))$
independent of the path, and target the resulting noisy probability estimate.
The noise-trader fraction is $\alpha \in [0,1]$.

\subparagraph{Information Structure.}
Each trader has access only to a local signal about a bounded number of events - there
is no global information. Informed traders observe the score process $S_t \in \mathbb{R}^M$
for all $M$ base markets, but this is a signal about the current conditional distribution
only; no trader has access to the full future path or to the joint distribution $P^*$
directly. PM itself must learn $P^*$ online from the sequence of incoming trades.

\subsection{Market Interface}
The market maker exposes a pricing and trading interface over the contract set $\mathcal{F}$. Trades are executed via a cost function $C_S(\cdot;\theta)$, which determines the payment required to adjust the position in contract $S$. Following each trade, the market maker updates its internal parameters $\theta$, which in turn updates prices across all contracts.

\subparagraph{Interaction model.}
Traders arrive sequentially. At each round $t$, one market $m_t$ is selected uniformly
at random from the $M + P$ markets, and a single trader arrives. With probability
$1 - \alpha$ the arriving trader is informed; with probability $\alpha$ they are a noise
trader. Competition is fair at each step: every arriving trader observes the current PM
price and decides whether and how much to trade before the next trader arrives.

\subsection{Design Objectives}

We formalize the requirements for a market maker supporting correlated contracts.


\subparagraph{Property 1: Liquidity Consistency.}
The market maker should provide comparable liquidity across all contracts without requiring independent capital allocation.

Let $I_m(x; \phi)$ denote the price impact of a trade $x$ in market $m$. Then there should exist constants $0 < c_1 \le c_2 < \infty$ such that for all sufficiently small admissible trade sizes $x$, such that there is non-zero price impact across all markets,

\begin{equation}
    c_1 \le \frac{I_m(x;\phi)}{I_{m'}(x;\phi)} \le c_2
\qquad \forall m,m' \in \mathcal{M}
\end{equation}
That is, price impact should remain comparable across all markets, including both base and parlay contracts, so that parlays are not disproportionately illiquid.

\subparagraph{Property 2: Controlled loss.}

 The market maker should control expected aggregate loss and tail risk under the induced joint process of trades and outcomes.

Let $L_T$ denote cumulative realized AMM loss over all executed trades up to horizon $T$.
A satisfactory mechanism should control both expected aggregate loss and tail risk under the induced joint process of trades and outcomes.

First, expected loss should scale subexponentially in the number of base markets, ideally linearly:
\begin{equation}
    \mathbb E[L_T] = o(2^M),
\qquad
\text{ideally } \mathbb E[L_T] = O(M)
\end{equation}

Second, tail loss should also remain controlled. 
For a confidence level $\alpha$ such as $0.95$, define
\begin{equation}
    \mathrm{VaR}_{\alpha}(L_T)
:=
\inf\{ \ell : \Pr(L_T \le \ell) \ge \alpha \}    
\end{equation}

Then the mechanism should satisfy
\begin{equation}
    \mathrm{VaR}_{0.95}(L_T) = o(2^M),
\qquad
\text{ideally } \mathrm{VaR}_{0.95}(L_T) = O(M)
\end{equation}
Optionally, the same requirement can be imposed on $\mathrm{CVaR}_{0.95}(L_T)$. 

\subparagraph{Property 3: Information aggregation}
The market maker should aggregate information from trades so that the implied distribution converges to the true distribution $P^{*}$.

Let $P_{\phi_T}$ denote the distribution represented by the AMM after $T$ trades.
A satisfactory mechanism should use singleton and parlay order flow to recover the latent joint law with small statistical error.

In particular, after a number of trades linear in the effective model dimension, the represented distribution should be close to the true distribution in KL divergence:
\begin{equation}
    D_{\mathrm{KL}}(P^\star \,\|\, P_{\phi_T}) \le \varepsilon
\qquad
\text{for } T = O(n),
\end{equation}
where $n$ denotes the relevant model dimension or number of informative parameters.
Equivalently, the KL error should decrease with the number of informative trades:
\begin{equation}
    \mathbb E\!\left[D_{\mathrm{KL}}(P^\star \,\|\, P_{\phi_T})\right]
\to 0
\quad \text{as } T \text{ grows.}
\end{equation}

\subparagraph{Implication for market design.}
Taken together, Properties~1--3 indicate that neither pure contract-by-contract execution nor pure contract-by-contract pricing is sufficient. Treating each exposed parlay as an isolated market tends to fragment liquidity and scale risk with the number of listed contracts, while quoting combinations purely as functions of current base-market prices can improve accessibility but does not, by itself, recover the latent joint law. The design problem is therefore to support executable parlay trading while pooling risk and information strongly enough to preserve liquidity, control aggregate loss, and use order flow to learn dependence. The next section examines two natural baseline approaches along these lines and shows where each falls short.


\section{Baseline Designs for Parlay-style Markets}
\label{sec:baseline}
The properties defined in Section \ref{sec:objectives} impose stringent requirements on any viable market design. We examine two natural baseline approaches - request-for-quote systems and independence-based automated market makers - and show that each fails to satisfy these requirements for fundamental, structural reasons. Additional empirical details, robustness checks, and full raw results are deferred to the appendix \ref{app:rfq-misprice}.

\subsection{Request-for-Quote Protocol}
\label{subsec:rfq}
In many existing platforms, parlay trades are handled via a request-for-quote (RFQ) mechanism: a trader specifies a bundle of outcomes, and a market maker returns a price based on proprietary models. This architecture is flexible, but it externalizes both pricing and risk
management to separate dealers. As a result, quotes are not necessarily derived
from a common joint distribution, and liquidity is fragmented across requested
combinations.


\subparagraph{Shortcomings of Using RFQ Protocol for Parlays}
This creates two structural failures. First, liquidity deteriorates for combinations: in our Kalshi combo sample, average spreads are substantially larger in combo markets than in base markets (Table~\ref{tab:base-vs-combo-spreads}). Second, RFQ quotes need not be probabilistically coherent. Since quotes are not marginals of a shared joint law, they can violate basic constraints such as the Fr\'echet bound \(P(A\cap B)\leq \min\{P(A),P(B)\); detailed violations are
reported in Appendix~\ref{app:rfq-misprice}. Aggregating quotes from multiple market makers yields only marginal improvements in spread, suggesting that fragmentation persists because competition among quote providers remains limited. A detailed breakdown by leg count and market is provided in Appendix~\ref{app:rfq-quote-agg}.

\subparagraph{The RFQ trilemma}
The RFQ architecture faces a fundamental trilemma among the three properties introduced in Section~\ref{sec:objectives}. Opening maker participation can improve quote availability and price discovery, but weakens loss control because risk is dispersed across less tightly capitalized makers. Conservative quoting protects makers from adverse selection, but the resulting quotes become poor probability estimates and therefore carry little information. Restricting the set of quoted combinations makes risk easier to manage, but sacrifices broad parlay liquidity. Thus, over a broad parlay family, RFQ systems can trade off among liquidity consistency, coherent information aggregation, and bounded loss, but cannot achieve all three simultaneously. We call this structural constraint the \emph{RFQ trilemma} (Figure~\ref{fig:rfq-trilemma}).

\begin{figure}[t]
\centering
\begin{tikzpicture}[
    x=1cm,
    y=1cm,
    every node/.style={font=\scriptsize},
    corner/.style={circle, fill=black, inner sep=1.2pt},
    prop/.style={font=\bfseries\scriptsize, align=center},
    sidebox/.style={
        rounded corners=3pt,
        draw=blue!25!black,
        fill=blue!4,
        minimum height=0.42cm,
        minimum width=2.35cm,
        align=center,
        inner sep=1.8pt,
        font=\scriptsize
    }
]

\coordinate (T) at (3.6,1.75);
\coordinate (L) at (0,0);
\coordinate (R) at (7.2,0);

\draw[line width=0.75pt] (L) -- (T) -- (R) -- cycle;

\node[corner] at (T) {};
\node[corner] at (L) {};
\node[corner] at (R) {};

\node[prop, above=3pt] at (T) {Liquidity\\consistency};
\node[prop, below=3pt, xshift=-0.15cm, align=center] at (L) {Information\\aggregation};
\node[prop, below=3pt, xshift=0.15cm, align=center] at (R) {Bounded\\loss};

\node[sidebox] at ($(L)!0.50!(T) + (-0.16,0.08)$) {Open participation\\low collateral};
\node[sidebox] at ($(T)!0.50!(R) + (0.16,0.08)$) {Conservative\\quotes};
\node[sidebox] at ($(L)!0.50!(R) + (0,-0.34)$) {Constrain to\\limited markets};

\end{tikzpicture}
\vspace{-0.8em}
\caption{The RFQ trilemma}
\label{fig:rfq-trilemma}
\vspace{-1.0em}
\end{figure}
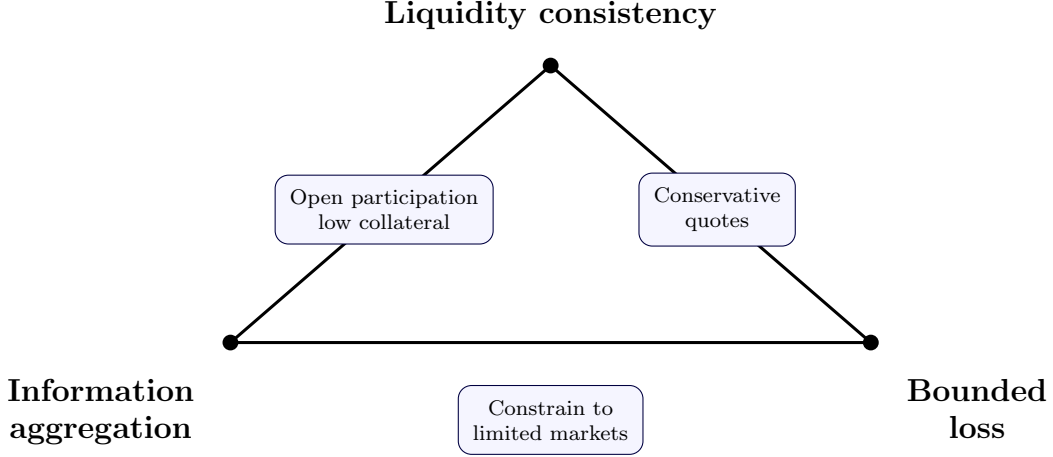

\begin{table}[t]
\centering
\caption{Average spreads in base and combo markets on Kalshi, restricting to markets with open interest between 1{,}000 and 2{,}000. Combo markets exhibit substantially wider spreads than base markets.}
\label{tab:base-vs-combo-spreads}
\begin{tabular}{lcc}
\toprule
Market type & Number of markets & Average spread (cents) \\
\midrule
Base markets  & 999 & 5.12 \\
Combo markets & 549 & 11.89 \\
\bottomrule
\end{tabular}
\end{table}

\subsection{Independence-Based Parlay AMM}
\label{subsec:ind-parlay-amm}

A second natural baseline is to keep independent AMMs for the base markets and quote a parlay as the product of its marginal prices:
\begin{equation}
    q_{\mathrm{ind}}(E_1 \cap \cdots \cap E_k) = \prod_{i=1}^k \hat p_i.
\end{equation} 
This gives immediate executable quotes, but it hard-codes independence and therefore fails whenever outcomes are correlated.

The failure is structural. Consider two positively correlated events with common marginal probability \(p\). The independence quote for the joint contract is \(p^2\), while the true joint probability can be close to \(p\). An informed trader can repeatedly buy the parlay at \(p^2\) against expected payoff \(p\). Because parlay trades do not update the base-market prices or any shared correlation state, the mispricing does not disappear through trading. The AMM therefore fails both bounded loss and information aggregation: it cannot learn the dependence structure no matter how much parlay order flow it observes.

\subsection{Towards a Joint Distribution Market Maker}
These observations suggest that any viable design must satisfy three structural requirements: (i) maintain a shared representation of the joint distribution, (ii) update this representation from both base and composite trades, and (iii) derive all prices coherently from this shared state. The next section introduces ParlayMarket, an AMM that satisfies these requirements by maintaining a tractable approximation to the joint distribution and updating it via online learning from trade flow.


\section{ParlayMarket}
\label{sec:correlation-mm}

The limitations of baseline designs in Section 4 arise from the absence of a shared global state and the inability to propagate information across markets. We now construct a market maker that resolves these limitations by maintaining and updating a coherent joint distribution over all outcomes.

At a high level, ParlayMarket interprets each trade - whether in a base market or a parlay - as a noisy observation about the underlying joint distribution $P^\star$. The mechanism maintains a parametric approximation $P_\phi$ to this distribution, derives all prices as marginals of $P_\phi$, and updates $\phi$ after every trade. This transforms the market into an online inference system, where prices reflect a globally consistent and continuously learned probabilistic model.

We show that this design satisfies all three properties from Section \ref{sec:objectives}, and that the complete-parlay variant enjoys an additional structural bonus: per-market loss halves with each new base asset added.


\subsection{Architecture}
\label{subsec:amm-architecture}

The central challenge for a multi-market maker is keeping prices across all
$N_M = 2^M - 1$ markets mutually consistent as information arrives through trade
flow.
A trade on any single market reveals something about the joint distribution, and
a rational AMM should propagate this information to related markets immediately - 
otherwise prices become stale and the AMM is exposed to arbitrage.
The ParlayMarket solves this through three interlocking mechanisms:
\emph{shadow trades} that broadcast information across markets,
an \emph{Ising model} belief state that compactly represents the joint distribution,
and a \emph{composite likelihood SGD} update that continuously calibrates the model
from observed trade flow.

\subparagraph{Shadow trades and price propagation.}
Prices in the ParlayMarket are posted marginals of a shared internal belief
distribution $P_{\boldsymbol{\varphi}}$ over all $M$ outcomes.
When a trader bets on market $S$, the observed trade direction reveals information
not just about event $S$ but about all correlated events.
Since we assume all incoming traders are informed (Section~\ref{sec:objectives}),
the AMM should treat each trade as a signal about the entire joint distribution and
update prices everywhere, not only in market $S$.

To implement this, the AMM records a set of internal \emph{shadow trades} on
related markets after each real trade arrives.
Shadow trades are not real bets - they are bookkeeping entries that encode the
information the AMM infers about other markets from the observed trade.
Critically, they ensure that no market's price becomes stale simply because it
receives little direct trading volume; even an esoteric parlay market that rarely
attracts traders will see its posted price updated whenever a correlated market
is active.

\subparagraph{Ising model as the joint belief state.}
To propagate information through shadow trades, the AMM needs a tractable model
of the joint distribution that can be updated incrementally.
Parlay trades are especially informative here: a bet on a $k$-way parlay
directly reveals a joint probability $P^*(\bigcap_{i \in S} \{X_i = 1\})$, which
constrains not only the individual event probabilities but also their pairwise and
higher-order correlations.

The Ising model is the natural parametric choice.
Given a target set of marginal probabilities $\{p_i^*\}$ and pairwise correlations
$\{p_{ij}^*\}$, the Ising model is the \emph{maximum-entropy} distribution
consistent with those first and second-order moments:
\begin{equation}\label{eq:ising-dist}
  P_{\boldsymbol{\varphi}}(\mathbf{x})
  \;\propto\;
  \exp\!\left(\sum_{i=1}^M \theta_i x_i
              + \sum_{i < j} W_{ij}\,x_i x_j\right),
\end{equation}
where $\boldsymbol{\varphi} = (\boldsymbol{\theta}, \mathbf{W})$ collects bias parameters
$\theta_i \in \mathbb{R}$ (controlling individual event probabilities) and
interaction weights $W_{ij} \in \mathbb{R}$ (controlling pairwise correlations).
Choosing the maximum-entropy distribution is equivalent to making the weakest
possible assumptions beyond the moment constraints: no spurious higher-order
structure is imposed that is not supported by the data.
This also corresponds to a moment-matching approach to Bayesian updating - rather
than tracking the full $2^M$-dimensional joint distribution, the AMM only needs
to maintain the $\binom{M}{2} + M$ sufficient statistics
$(\mathbb{E}[X_i], \mathbb{E}[X_i X_j])$ of the Ising exponential family.

Posted prices are the Ising marginals computed via belief propagation:
$p_i^{\boldsymbol{\varphi}} = P_{\boldsymbol{\varphi}}(X_i = 1)$ for base markets and
$p_S^{\boldsymbol{\varphi}} = P_{\boldsymbol{\varphi}}\!\bigl(\bigcap_{i \in S}\{X_i = 1\}\bigr)$
for a parlay $S \subseteq [M]$.
All $N_M$ prices are generated from the single shared parameter vector
$\boldsymbol{\varphi}$, so they are guaranteed to be globally consistent with some
joint distribution at every point in time.

\subparagraph{Gradient updates from trade errors.}
Each trade on market $m$ reveals a discrepancy between the AMM's current posted
price $p_m^{\boldsymbol{\varphi}}$ and the true probability $p_m^* = P^*(\mathcal{E}_m)$
implied by the trader's signal.
The AMM treats this discrepancy as a cross-entropy error and takes a stochastic
gradient step to reduce it.
Formally, the AMM implicitly minimises the \emph{composite pseudo-likelihood}~\cite{besag1975statistical,lindsay1988composite}:
\begin{equation}\label{eq:pseudo-loss}
  \mathcal{L}(\boldsymbol{\varphi})
  \;=\;
  \sum_{i=1}^M \lambda_i\,\mathrm{CE}(p_i^*, p_i^{\boldsymbol{\varphi}})
  \;+\;
  \sum_{i < j} \lambda_{ij}\,\mathrm{CE}(p_{ij}^*, p_{ij}^{\boldsymbol{\varphi}}),
\end{equation}
where $\lambda_i, \lambda_{ij} \ge 0$ are normalised trade arrival rates and
$\mathrm{CE}(p,q) = -p\log q - (1-p)\log(1-q)$.
Upon each trade on market~$m$, the AMM applies the SGD update
\begin{equation}\label{eq:sgd-update}
  \boldsymbol{\varphi}_{t+1}
  \;=\;
  \boldsymbol{\varphi}_t
  - \eta\,\nabla_{\boldsymbol{\varphi}}\,\mathrm{CE}(p_m^*, p_m^{\boldsymbol{\varphi}_t}),
\end{equation}
with step size $\eta > 0$.
Because $\boldsymbol{\varphi}$ is shared across all markets, a gradient step triggered
by a trade on market $m$ immediately shifts the posted prices of every other
market through the Ising model's interaction terms.
Shadow trades are then computed using the updated $\boldsymbol{\varphi}_{t+1}$, so the
propagated price changes are themselves grounded in the newly calibrated belief state.
This cycle - real trade $\to$ gradient update $\to$ shadow trades - is the core
feedback loop of the ParlayMarket.

\subparagraph{Relation to Properties~1--3.}
The three mechanisms above address our desiderata from \prettyref{sec:objectives} in turn. We now give an informal preview of the reasons behind this, and then prove formal results and empirical justification in the succeeding sections.

\subparagraph{Property~1} (Liquidity consistency) is satisfied because
all $N_M$ posted prices flow from a single shared $\boldsymbol{\varphi}$, there
is no separate model or order book per market.
More importantly, shadow trades mean the AMM does not need direct trading
activity in a market to keep its price current: every real trade on any
correlated market triggers updates that ripple through $\boldsymbol{\varphi}$ and
shift all posted prices simultaneously.
A rarely-traded parlay market is not forgotten; it benefits from every trade
on its correlated counterparts. Complete proof is deferred to \prettyref{app:proof-prop-1}.

\subparagraph{Property~2} (Bounded loss) follows from the balance that shadow
trades create between gradient signal and gradient noise.
Each real trade adds one update to $\boldsymbol{\varphi}$, but the shadow trades it
spawns spread that update's effect across all related markets.
As the number of markets grows, both the number of updates per round and the
noise in each update grow together and cancel, keeping the steady-state
parameter error, and therefore the AMM's total monetary loss per round, bounded independently of $M$.
Without shadow trades this cancellation breaks, and loss grows exponentially
with $M$.

\subparagraph{Property~3} (Information aggregation) follows from the
structure of the cross-entropy objective.
The composite pseudo-likelihood $\mathcal{L}(\boldsymbol{\varphi})$ has a unique
minimum at the parameters $\boldsymbol{\varphi}^*$ whose implied moments match those
of the true distribution $P^*$.
Because each trade gradient pushes $\boldsymbol{\varphi}$ toward this minimum,
the AMM's posted prices converge to those of the best Ising approximation
to $P^*$ as trade flow accumulates, and with it, the posted joint
probabilities approach the true joint distribution.

\subsection{Convergence to the True Distribution}
\label{subsec:convergence}
We first show how Property~3 is satisfied by focusing on how well the Ising model learns the joint distribution via the incoming trades.
Our main result concerns the posted prices of the ParlayMarket - they converge geometrically to the best Ising
approximation of $P^*$. After $O(1/(\eta\mu))$ informative trades the KL divergence
between the true joint distribution and the AMM's belief shrinks to the model-misspecification
floor, and all $N_M$ posted prices stabilize at their moment-matched values.

\begin{theorem}[Geometric convergence]
\label{thm:linear-conv}
In the pure arbitrageur environment, starting from $\boldsymbol{\varphi}_0 = 0$ with step
size $\eta \le 1/(2L)$:
\[
  \mathbb{E}\|\boldsymbol{\varphi}_t - \boldsymbol{\varphi}^*\|^2
  \;\le\;
  (1 - 2\eta\mu)^t\,\|\boldsymbol{\varphi}^*\|^2
  \;+\;
  \frac{\eta\sigma^2}{2\mu},
\]
where $\sigma^2 = \mathbb{E}\|\mathbf{g}_t(\boldsymbol{\varphi}^*)\|^2$ is the gradient variance at
the optimum.  With $\eta = 0.2$ and $\mu \approx \tfrac{1}{4}$, the contraction
factor per round is $0.9$; roughly $22$ rounds suffice to reduce the initial error
by $10\times$.
\end{theorem}

Three results underpin this guarantee.

\subparagraph{Moment-matched target.}
The fixed point $\boldsymbol{\varphi}^*$ is the unique parameter vector satisfying
$p_i^{\boldsymbol{\varphi}^*} = p_i^*$ for all $i$ and $p_{ij}^{\boldsymbol{\varphi}^*} = p_{ij}^*$
for all $i < j$.
This follows directly from the update rule: at $\boldsymbol{\varphi}^*$ the per-trade scalar
$\lambda_m = (p_m^{\boldsymbol{\varphi}^*} - p_m^*) / (1 - p_m^{\boldsymbol{\varphi}^*})$ vanishes for every
market $m$, so all gradient steps are zero.
Equivalently, $\boldsymbol{\varphi}^*$ is the I-projection of $P^*$ onto the Ising family - the
maximum-entropy distribution whose first and second moments match those of
$P^*$~\cite{wainwright2008graphical}.
Once $\boldsymbol{\varphi}_t \to \boldsymbol{\varphi}^*$, the posted joint prices
$p_S^{\boldsymbol{\varphi}} = P_{\boldsymbol{\varphi}}(\bigcap_{i \in S}\{X_i=1\})$ converge
to the best Ising-family approximation of $P^*$, giving
$D_{\mathrm{KL}}(P^* \| P_{\boldsymbol{\varphi}_t}) \to 0$ up to the misspecification gap from higher-order dependence not captured by the pairwise family.

\subparagraph{Loss curvature.}
The contraction rate $(1 - 2\eta\mu)$ is determined by the strong-convexity
constant $\mu$ of the composite pseudo-likelihood $\mathcal{L}$.

\begin{proposition}[Strong convexity]
\label{prop:strong-convex}
On the compact ball $\Phi_B = \{\boldsymbol{\varphi} : \|\boldsymbol{\varphi}\| \le B\}$,
$\mathcal{L}$ is $\mu$-strongly convex.
At the uninformative initialisation $\boldsymbol{\varphi} = 0$, the Hessian is diagonal
with $\mu = \tfrac{3}{16}$ (interaction directions) and $\Lambda = \tfrac{1}{4}$
(bias directions), giving condition number $\kappa = \Lambda/\mu = \tfrac{4}{3}$.
All eigenvalues are $O(1)$, independent of $M$.
\end{proposition}

Each CE term in $\mathcal{L}$ contributes a positive-semidefinite rank-one block
$H_m = (\nabla_{\boldsymbol{\varphi}} p_m^{\boldsymbol{\varphi}})(\nabla_{\boldsymbol{\varphi}} p_m^{\boldsymbol{\varphi}})^\top / [p_m^*(1-p_m^*)]$
to the Hessian via the covariance identity.
At $\boldsymbol{\varphi} = 0$ the Ising model is the uniform product measure, so cross-terms vanish
and the Fisher matrix is diagonal with entries $\tfrac{1}{4}$ (bias parameters) and
$\tfrac{3}{16}$ (interaction parameters) - all $O(1)$ regardless of $M$.
Adding more base markets adds more CE terms and therefore more curvature, but the
per-parameter curvature stays bounded, which is why the convergence rate does not
degrade as $M$ grows.

\subparagraph{Parlay acceleration.}
Higher-order parlay trades strictly improve the convergence rate by adding curvature in
the interaction directions that base-market trades cannot supply alone.
When the AMM accepts $k$-way parlay bets for $k \ge 3$, each parlay-$S$ trade
contributes a further PSD block $H_S$ to the Hessian, raising the effective
strong-convexity constant to
\begin{equation}\label{eq:mu-eff}
  \mu_\mathrm{eff}
  \;\ge\; \mu + \sum_{k \ge 3}\;\sum_{S \in \mathcal{S}_k} \lambda_S\,\mu_S,
  \qquad
  \mu_S = \lambda_{\min}(H_S) \ge 0,
\end{equation}
with strict inequality whenever parlay $S$ is informative about
$\boldsymbol{\varphi}$.
The contraction factor improves from $(1 - 2\eta\mu)$ to $(1 - 2\eta\mu_\mathrm{eff})$,
and the cumulative regret tightens to $G^2/(2\mu_\mathrm{eff})(1+\log T) = O(\log T)$.
The gain is largest for the interaction weights $W_{kl}$: a two-way parlay $\{k,l\}$
contributes only $O(\rho^2)$ to the $(W_{kl},W_{kl})$ curvature entry via base-price
cross-effects, whereas a higher-order parlay $S \ni k,l$ contributes $O(1)$ directly - 
a factor of $O(\rho^{-2})$ improvement at small correlations $\rho$.
Proofs are deferred to the appendix.

\subsection{Loss Bounds}
\label{subsec:per-market-decay}

The previous subsection established Property~3 by showing that the shared parameter vector
$\boldsymbol{\varphi}_t$ converges to the moment-matched target
$\boldsymbol{\varphi}^\ast$. We now turn to Property~2. The question here is not whether the
AMM learns the correct joint distribution, but how the remaining parameter error translates into
monetary loss. The bridge is a quadratic sensitivity bound: near the fixed point, LMSR loss is
second-order in pricing error, and pricing error is first-order in parameter error. In the
complete-parlay regime, shadow trades then pool information across all
\[
  N_M = 2^M - 1
\]
active contracts, keeping the \emph{aggregate} steady-state loss per round bounded even though
the number of quoted markets grows exponentially.

\begin{proposition}[Complete-parlay loss scaling]
\label{prop:per-market-decay}
In the complete-parlay regime, where every non-empty subset $S \subseteq [M]$ is traded at equal
rate and the SGD step size is $\eta$, the aggregate expected LMSR loss admits the decomposition
\[
  \mathbb{E}[L_T]
  \;\le\;
  L_{\mathrm{transient}} + T\,L_{\mathrm{ss,round}},
\]
with
\[
  L_{\mathrm{transient}}
  =
  O\!\left(\frac{b\rho^2 M^2}{\eta}\right),
  \qquad
  L_{\mathrm{ss,round}}
  =
  O(b\eta).
\]
In particular, the dependence on $M$ is polynomial rather than exponential, so the
complete-parlay Parlay AMM satisfies the bounded-loss requirement of Property~2.
\end{proposition}

The proposition should be read as a two-phase decomposition. The first term is a one-time
learning cost incurred while the AMM moves from the uninformative initialization toward
$\boldsymbol{\varphi}^\ast$. The second is the steady-state cost of continuing to learn from a
stochastic stream of informative trades once the model has essentially converged.

\subparagraph{Loss--parameter link.}
We first convert the parameter-space error from Theorem~\ref{thm:linear-conv} into a bound on
monetary loss.

\begin{lemma}[Quadratic sensitivity]
\label{lem:quad-sensitivity}
For a market $S$ with true probability $p_S^\ast$ and AMM parameters
$\boldsymbol{\varphi}$, the expected LMSR monetary loss from one informed trade satisfies
\[
  b\,\mathrm{KL}(p_S^\ast \| p_S^{\boldsymbol{\varphi}})
  \;\le\;
  \frac{b\,\Lambda_S}{2}\,
  \|\boldsymbol{\varphi}-\boldsymbol{\varphi}^\ast\|^2,
  \qquad
  \Lambda_S
  :=
  \sup_{\boldsymbol{\varphi}}
  \frac{\|\nabla_{\boldsymbol{\varphi}} p_S^{\boldsymbol{\varphi}}\|^2}
       {p_S^\ast(1-p_S^\ast)}.
\]
\end{lemma}

The lemma is the composition of two local facts. First, for $p$ near $p^\ast$,
\[
  \mathrm{KL}(p^\ast\|p)
  =
  \frac{(p-p^\ast)^2}{2p^\ast(1-p^\ast)}
  + O(|p-p^\ast|^3).
\]
Second, the price error is linear in the parameter error:
\[
  p_S^{\boldsymbol{\varphi}} - p_S^\ast
  =
  \nabla_{\boldsymbol{\varphi}} p_S^{\boldsymbol{\varphi}^\ast}
  \cdot
  (\boldsymbol{\varphi}-\boldsymbol{\varphi}^\ast)
  + O(\|\boldsymbol{\varphi}-\boldsymbol{\varphi}^\ast\|^2).
\]
Combining the two shows that monetary loss is quadratic in parameter error. This is the key
bridge from the convergence theory to Property~2.

\subparagraph{Signal--noise balance in the complete-parlay regime.}
In the formal interaction model, one market is selected uniformly at random each round. Over a
block of $N_M$ rounds, each of the $N_M = 2^M - 1$ markets is therefore visited once
\emph{in expectation}. The cumulative local information gathered over such a block is summarized by
the aggregate sensitivity matrix
\[
  \Sigma_{\mathrm{cycle}}
  :=
  \sum_{S\neq\emptyset}
  F_S(\boldsymbol{\varphi}^\ast),
  \qquad
  F_S(\boldsymbol{\varphi}^\ast)
  :=
  \frac{
    (\nabla_{\boldsymbol{\varphi}} p_S^\ast)
    (\nabla_{\boldsymbol{\varphi}} p_S^\ast)^\top
  }{
    p_S^\ast(1-p_S^\ast)
  }.
\]
Each term $F_S(\boldsymbol{\varphi}^\ast)$ measures how informative market $S$ is about the shared
parameter vector near the truth: markets whose prices respond strongly to perturbations of
$\boldsymbol{\varphi}$ contribute more to learning. Because all markets update the same shared
state $\boldsymbol{\varphi}$, both the mean restoring drift toward
$\boldsymbol{\varphi}^\ast$ and the stochastic gradient noise accumulate through this same pooled
matrix $\Sigma_{\mathrm{cycle}}$. Put differently, adding more markets increases signal and noise
at the same rate rather than letting noise dominate. As a result, the steady-state parameter error
remains $O(\eta)$ rather than growing with $N_M$. Applying
Lemma~\ref{lem:quad-sensitivity} then gives
\[
  L_{\mathrm{ss,round}} = O(b\eta).
\]

\subparagraph{Why shadow trades matter.}
This favorable scaling is not automatic: it relies on the fact that a trade in one market updates
the shared parameter vector and therefore refreshes all related quotes. That cross-market
propagation is exactly what shadow trades provide. A trade in market $S$ does not only improve the
quote for $S$; it also carries information about overlapping base and parlay contracts through the
common state $\boldsymbol{\varphi}$. Without this mechanism, each additional market would still add
its own stochastic noise, but the corresponding information would remain largely confined to the
market that was directly traded. The system would then behave more like a collection of weakly
connected books than a single pooled market maker, and the uniform aggregate bound in
Proposition~\ref{prop:per-market-decay} would no longer be expected to hold.

As an immediate corollary, once the transient learning phase has passed,
\[
  \bar{\ell}_M
  \;\approx\;
  \frac{L_{\mathrm{ss,round}}}{N_M}
  =
  O\!\left(\frac{b\eta}{2^M}\right),
\]
so the \emph{mean loss per listed market} decays exponentially with $M$. We state this as a
corollary rather than part of Proposition~\ref{prop:per-market-decay} because the design objective
in Section~\ref{sec:objectives} is bounded \emph{aggregate} loss; the halving pattern in Figure~\ref{fig:per-market-with-parlay} and Table~\ref{tab:per-market-ratios} is the
complete-parlay consequence of that aggregate bound.

\subsection{Information-theoretic optimality}
\label{subsec:optimality}

The quadratic loss scaling in Proposition~\ref{prop:per-market-decay} is worst-case optimal. To see why, consider the pairwise family
\[
    P_s(x) \propto
    \exp\!\Big(\gamma \sum_{i<j} s_{ij} x_i x_j\Big),
    \qquad
    s_{ij}\in\{-1,+1\}.
\]
This family contains one independent hidden correlation sign for every pair of events, hence \(\binom{M}{2}=\Theta(M^2)\) independent pairwise facts. Any AMM that does not know these signs in advance must learn them from order flow in order to price the corresponding two-leg parlays correctly.

For a fixed pair \((i,j)\), the two possible signs \(s_{ij}=+1\) and \(s_{ij}=-1\) imply different fair prices for the parlay \(X_iX_j\). Until the sign is learned, any AMM quoting this parlay with liquidity scale \(b\) must be wrong on one of the two cases, and an informed trader can extract loss from that mispricing. For constant-strength correlations this loss is \(\Omega(b)\) per unresolved pair; for weak correlations of size \(\gamma\) it is \(\Omega(b\gamma^2)\).

Since there are \(\Theta(M^2)\) independent pairwise signs, even the best AMM in this class has worst-case information-acquisition loss
\[
    \inf_A \sup_s \mathbb{E}_{P_s}[L_T]
    \;\ge\;
    \Omega\!\left(b{M \choose 2}\right)
    =
    \Omega(bM^2),
\]
for constant-strength correlations, and \(\Omega(b\gamma^2M^2)\) in the weak-correlation regime. Thus the quadratic loss rate is not an artifact of the Ising update rule: it is the unavoidable cost of learning dense pairwise correlation structure from trades. \textit{ParlayMarket} matches this rate up to constants while avoiding the exponential loss of treating each parlay as an independent market.

\subsection{Parlay Trades as essential for learning}
\label{subsec:parlay-acceleration}

The previous subsection used Theorem~\ref{thm:linear-conv} only as an input to the loss bound.
We now return to the learning dynamics themselves. Higher-order parlays are not merely extra
contracts to quote; they also sharpen the learning problem by adding curvature to the shared
objective. This strengthens Property~3 directly, and through
Lemma~\ref{lem:quad-sensitivity} it improves the constants in Property~2 as well.

\begin{proposition}[Parlay-accelerated convergence]
\label{prop:parlay-acceleration}
Suppose, in addition to base and pairwise markets, the AMM receives trades on higher-order
parlays $S \subseteq [M]$ with $|S|\ge 3$ at rates $\{\lambda_S\}$. Let $\mathcal{L}^+$ denote the
augmented composite loss including these parlay terms, and write
\[
  \mathcal{S}_k := \{S \subseteq [M] : |S| = k\}.
\]
Then:
\begin{enumerate}
  \item \textbf{Improved local curvature.}
  In a neighborhood of $\boldsymbol{\varphi}^\ast$, the augmented loss has local
  strong-convexity constant $\mu_{\mathrm{eff}}$ satisfying
  \[
    \mu_{\mathrm{eff}}
    \;\ge\;
    \mu + \sum_{k \ge 3}\;\sum_{S \in \mathcal{S}_k} \lambda_S\,\mu_S,
    \qquad \mu_S \ge 0,
  \]
  with strict inequality whenever at least one traded higher-order parlay is informative about
  $\boldsymbol{\varphi}$.

  \item \textbf{Faster parameter convergence.}
  Under the same SGD dynamics~\eqref{eq:sgd-update} with step size
  $\eta \le 1/(2L)$,
  \[
    \mathbb{E}\|\boldsymbol{\varphi}_t - \boldsymbol{\varphi}^\ast\|^2
    \;\le\;
    (1 - 2\eta\mu_{\mathrm{eff}})^t\,\|\boldsymbol{\varphi}^\ast\|^2
    \;+\;
    \frac{\eta\sigma^2}{2\mu_{\mathrm{eff}}},
  \]
  yielding strictly faster geometric contraction than under the base-and-pairwise objective alone.

  \item \textbf{Improved average error and loss.}
  Let
  \[
    \bar{\varepsilon}_T^2
    :=
    \frac{1}{T}\sum_{t=1}^T
    \mathbb{E}\|\boldsymbol{\varphi}_t - \boldsymbol{\varphi}^\ast\|^2.
  \]
  Then
  \[
    \bar{\varepsilon}_T^2
    \;\le\;
    O\!\left(
      \frac{\|\boldsymbol{\varphi}^\ast\|^2}{\eta\mu_{\mathrm{eff}} T}
      + \eta
    \right).
  \]
  Consequently, the average LMSR loss per round is
  $O(b\,\bar{\varepsilon}_T^2)$ and the cumulative excess loss over $T$ rounds is
  $O(bT\,\bar{\varepsilon}_T^2)$, both with strictly improved constants.
\end{enumerate}
\end{proposition}

\subparagraph{Intuition: curvature as information aggregation.}
Each trade provides a stochastic gradient whose strength depends on how informative the observed
event is about the parameters. Base-market trades mainly constrain first moments, and pairwise
trades constrain interactions only one edge at a time. By contrast, a $k$-way parlay directly
reveals a joint probability and therefore aligns with several interaction directions simultaneously.
At the level of the Hessian, each informative parlay adds a positive semidefinite block, strengthening
the restoring force toward $\boldsymbol{\varphi}^\ast$.

\subparagraph{Interaction-level acceleration.}
The gain is largest for interaction weights $W_{ij}$. In weak-correlation regimes, these directions
are the hardest to learn from singleton flow alone, because the relevant gradients appear only
through small cross-effects. Higher-order parlays that include both $i$ and $j$ constrain several
interaction directions at once and can therefore dominate the curvature in the weakly identified
subspace; heuristically, the improvement can be of order $O(\rho^{-2})$ when $\rho \ll 1$.

\subparagraph{From faster learning to lower loss.}
Proposition~\ref{prop:parlay-acceleration} should be read as the mechanism behind the favorable
loss scaling in Proposition~\ref{prop:per-market-decay}. By
Lemma~\ref{lem:quad-sensitivity}, monetary loss is quadratic in parameter error, so any increase
in curvature reduces both the transient learning cost and the steady-state error floor. Parlays are
therefore not only additional products to quote; they are an additional source of statistical signal
that sharpens the learning dynamics of the market maker itself.


\section{Evaluation}
\label{sec:evaluation}

We evaluate ParlayMarket in both controlled and market-based settings. 
The synthetic experiments test the theoretical predictions of Section~5: bounded aggregate loss, geometric convergence, and the role of parlay order flow in accelerating learning. 
We then study a historical-market replay on Kalshi data to evaluate whether the same advantages persist under observed prices, RFQ flow, and realistic implementation constraints.


\subsection{Complete-Parlay Loss Scaling}
\label{subsec:eval-complete-parlays}

\subparagraph{Setup.}
We simulate $M \in \{4, 5, 6, 7, 8, 9\}$ correlated binary assets drawn from a
jointly Gaussian scoring model with pairwise correlation $\rho = 0.3$.
At each round, one of the $N_M = 2^M - 1$ non-empty parlay markets is selected
uniformly at random and an informed trader submits a bet of size $b$.
The AMM updates its Ising parameters via SGD with step sizes $\eta_h = \eta_J = 0.2$.
All results are averaged over $100$ independent simulation runs.

\subparagraph{Baselines.}
We compare the Correlation AMM against four baselines.
The \emph{true oracle} knows $P^*$ exactly and serves as a lower bound on achievable loss.
The \emph{independent LMSR} maintains a separate LMSR per market, ignoring
correlations entirely.
The \emph{pairwise oracle} knows the true pairwise marginals $p_{ij}^*$ analytically
and fits an Ising MRF to these moments, but does not learn from trade flow.
The \emph{Gaussian oracle} uses the true Gaussian copula parameters to price
joint events via multivariate normal CDF; it captures higher-order structure
exactly under the Gaussian model but cannot update from trades.
All oracle baselines are given their parameters for free and serve to
bracket the loss achievable with varying degrees of distributional knowledge.

\subparagraph{Per-market loss decay.}
Figure~\ref{fig:per-market-with-parlay} shows the mean per-market LMSR monetary loss
$\ell_M$ as a function of $M$.
Consistent with Proposition~\ref{prop:per-market-decay}, $\ell_M$ decays
exponentially as $C \cdot 2^{-M}$, halving with each additional base asset.
The product $\ell_M \cdot 2^M \approx 0.15$ is approximately constant for
$M = 4, \ldots, 9$ (see Table~\ref{tab:per-market-ratios} in Appendix \prettyref{app:plots}), confirming that
the total per-round loss $N_M \cdot \ell_M$ is essentially flat---consistent
with the steady-state dominance predicted by the two-phase decomposition.

\begin{figure}[t]
  \centering
  \begin{subfigure}[t]{0.49\columnwidth}
    \includegraphics[width=\columnwidth]{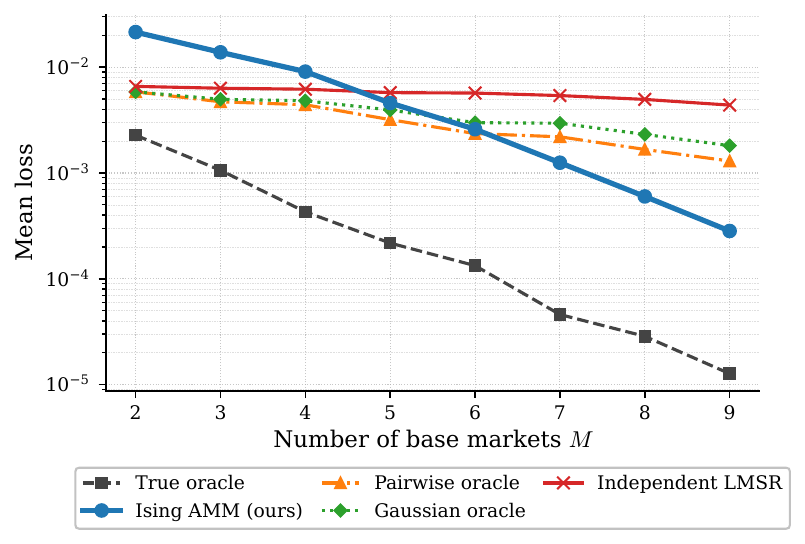}
    \caption{Mean LMSR loss per round per market $\bar{\ell}_M$.}
    \label{fig:mean-per-market-with-parlay}
  \end{subfigure}
  \hfill
  \begin{subfigure}[t]{0.49\columnwidth}
    \includegraphics[width=\columnwidth]{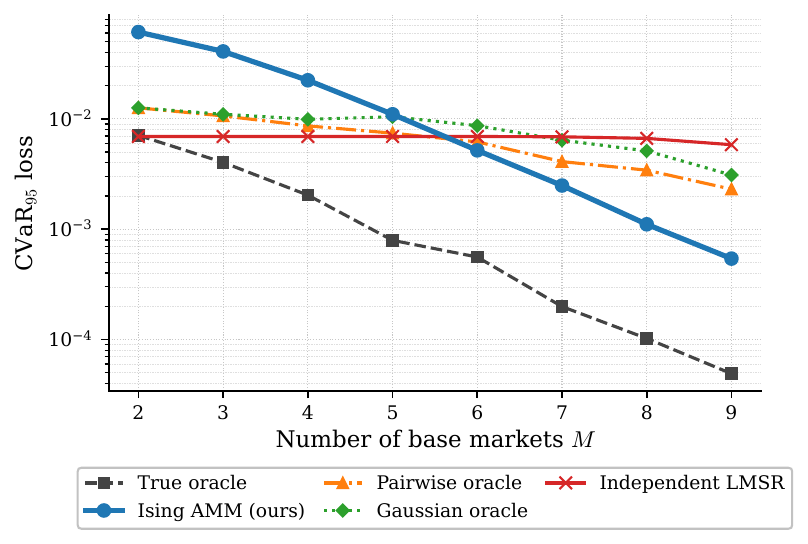}
    \caption{$\mathrm{CVaR}_{95}$ loss per round per market.}
    \label{fig:cvar-per-market-with-parlay}
  \end{subfigure}
  \caption{The ParlayMarket is the only non-oracle model whose loss decays exponentially with the number of base markets $M$, in the presence of parlays; all other non-oracle models plateau or grow.}
  \label{fig:per-market-with-parlay}
\end{figure}

\subparagraph{Price convergence.}
Figure~\ref{fig:price-convergence} in Appendix~\prettyref{app:plots} shows the mean absolute error (MAE)
between the AMM's posted prices and the true probabilities over time,
averaged across all $N_M$ markets and all simulation runs, with $\pm 1\sigma$
shading.
The top panel shows the LMSR market-price MAE; the bottom panel shows the
Ising marginal belief error for the Ising model.
Both converge monotonically; the Ising belief error converges faster
because belief propagation directly optimises the Ising objective,
while the market-price error includes the additional LMSR rounding effect.

\subparagraph{Parlay acceleration.}
Consistent with Proposition~\ref{prop:parlay-acceleration}, the advantage of the ParlayMarket
is entirely contingent on the presence of parlay order flow.
Figure~\ref{fig:per-market-with-parlay} shows that with complete parlays the Ising
AMM achieves the lowest non-oracle loss, decaying exponentially with $M$.
Figure~\ref{fig:per-market-no-parlay} shows the complementary picture:
when informed traders are restricted to base markets only, the model ranking
completely reverses.

\begin{figure}[t]
  \centering
  \begin{subfigure}[t]{0.48\columnwidth}
    \includegraphics[width=\columnwidth]{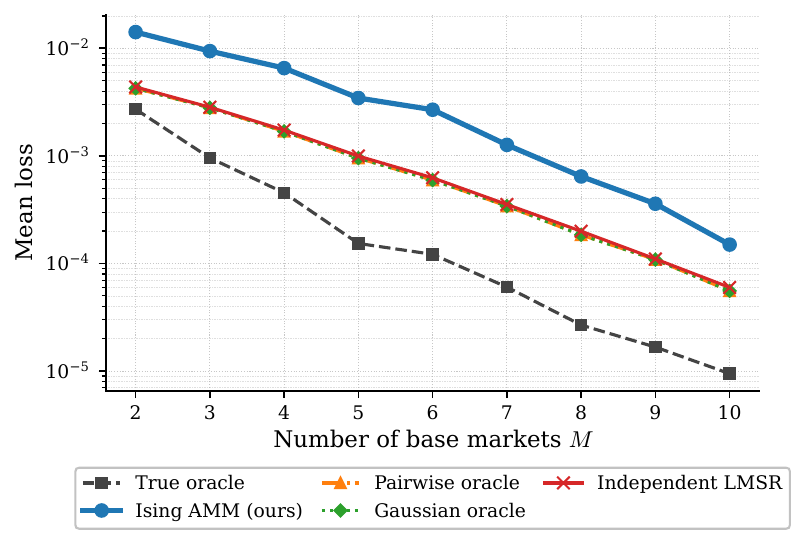}
    \caption{Mean LMSR loss per round per market $\bar{\ell}_M$.}
    \label{fig:mean-per-market-no-parlay}
  \end{subfigure}
  \hfill
  \begin{subfigure}[t]{0.48\columnwidth}
    \includegraphics[width=\columnwidth]{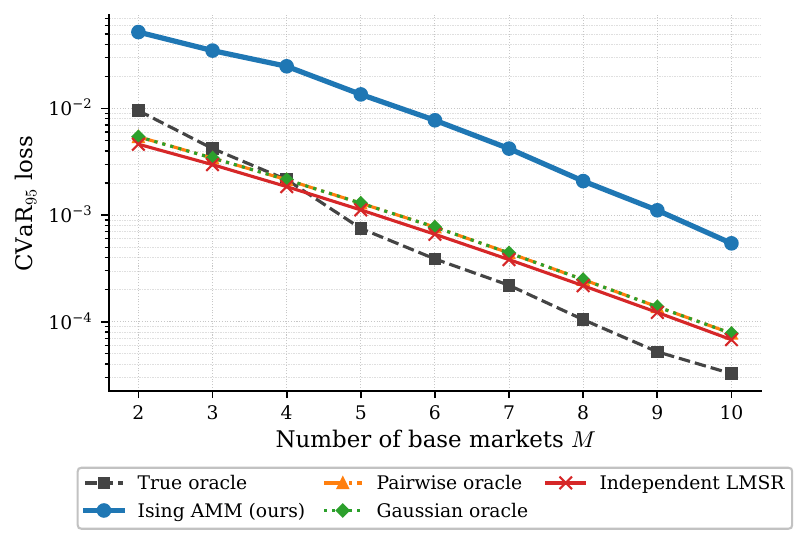}
    \caption{$\mathrm{CVaR}_{95}$ loss per round per market.}
    \label{fig:cvar-per-market-no-parlay}
  \end{subfigure}
  \caption{The ParlayMarket fails to perform better than even independent AMMs in the absence of parlay markets.}
  \label{fig:per-market-no-parlay}
\end{figure}

\subsection{Importance of Parlays for learning}
\label{subsec:eval-no-parlay}

\subparagraph{Setup.}
We run an ablation in which informed traders are restricted to the $M$ base
markets only - no parlay bets are placed - while all AMMs still maintain and
price all $N_M = 2^M - 1$ markets.
Figures~\ref{fig:per-market-with-parlay} and~\ref{fig:per-market-no-parlay}
show mean loss per round per market vs.\ $M$ for both conditions.

\subparagraph{The model ranking reverses completely.}
With complete parlays (Figure~\ref{fig:per-market-with-parlay}), the ParlayMarket
is the \emph{best-performing} non-oracle model: it achieves the lowest per-round
loss of all non-oracle AMMs, sitting just above the true oracle and well below every
other model.
The independent LMSR, pairwise oracle, and Gaussian oracle all incur
substantially higher and exponentially growing losses as $M$ increases, reaching
$\approx 2.24$, $\approx 0.66$, and $\approx 0.93$ per round respectively at $M=9$.

Without parlays (Figure~\ref{fig:per-market-no-parlay}), the ranking
\emph{completely reverses}: the ParlayMarket now has the \emph{highest} loss of all
non-oracle models, while the independent LMSR, pairwise oracle, and Gaussian oracle
cluster together at much lower values ($\approx 0.02$--$0.06$ per round across
all tested $M$).
The y-axis scale tells the story: the no-parlay plot peaks at $\approx 0.18$,
while the with-parlay plot requires a scale of $2.3$ - a $13\times$
difference - just to accommodate the losses of the non-Ising models.

\subparagraph{Why the ranking reverses.}
The reversal reveals precisely what parlays provide: the \emph{gradient signal
needed to calibrate the Ising joint belief model}. Without parlays, the ParlayMarket receives gradient updates only from the $M$ base
markets.
This is insufficient to pin down the $\binom{M}{2}$ interaction parameters
$W_{ij}$ in any reasonable simulation length: the base-market signal is
under-determined relative to the full Ising parameter space, so the model
misprices joint events and incurs high loss.
Simpler models - the pairwise oracle (which knows $p_{ij}^*$ analytically) and
the independent LMSR (which only needs to price $M$ uncorrelated events) - do not
face this calibration problem and therefore achieve low loss without parlays. With parlays, the situation reverses.
Each parlay bet supplies a direct gradient signal on a specific joint
event, giving the Ising model the rich information it needs to
calibrate $(\boldsymbol{\theta}, \mathbf{W})$ from all $N_M$ markets simultaneously
(Proposition~\ref{prop:parlay-acceleration} and Figure~\ref{fig:per-market-with-parlay}).
Once calibrated, belief propagation ensures that all $N_M$ posted prices are
globally consistent with the learned joint distribution, so adversaries cannot
exploit mispriced correlations.
The independent LMSR cannot use this information at all (it ignores
correlations by design), and the pairwise and Gaussian oracles only capture
lower-order structure - leaving them vulnerable to adversaries who trade
higher-order parlays against their mispriced joint events.

\subparagraph{Implication.}
Parlays are not merely an efficiency accelerant: they are the \emph{primary
information channel} through which the ParlayMarket learns the joint distribution.
Without them, the Ising model is the worst non-oracle performer; with them, it
is the best.
No other model tested achieves this reversal, confirming that the ParlayMarket
and parlay markets are jointly necessary - neither is sufficient alone.

\subsection{Robustness Under Noisy Informed Traders}
\label{subsec:eval-noise}

\subparagraph{Setup.}
We repeat the complete-parlay simulation with the same parameters but vary
the noise-trader fraction $\alpha \in [0, 1]$: with probability $\alpha$ the
arriving trader submits a uniformly random price signal, and with probability
$1 - \alpha$ they are fully informed.
Results at $M = 9$ are extrapolated from direct simulations at
$M \in \{3, 4, 5\}$ using the clean $M = 5 \to 9$ loss scaling.

\subparagraph{Results.}
Figure~\ref{fig:noise-loss} in Appendix \prettyref{app:plots} shows the per-market LMSR loss as a function of
noise-trader fraction $\alpha$ for $M = 9$.
The ParlayMarket's loss grows gracefully with $\alpha$---it remains the
best-performing non-oracle model across all noise levels tested---because
the noisy gradient signal inflates the steady-state error while leaving the
convergence rate largely unchanged.
All models cross the break-even line only at high noise fractions,
demonstrating robustness to realistic noise levels.

\subsection{Extension to Multinomial Markets}
\label{subsec:eval-multinomial}

\subparagraph{Setup.}
All experiments above assume binary outcomes per base asset.
We now evaluate whether the Ising AMM generalises to \emph{multinomial} markets,
where each base asset can resolve to one of $K > 2$ outcomes.
We run three configurations that vary the outcome structure and the
nature of price jumps (e.g.\ smooth updates vs.\ discontinuous jumps between
outcome clusters).
In each case, we compare the Ising AMM against a \emph{True Oracle} baseline
(knows the true joint distribution exactly) and an \emph{Independent} baseline
(maintains a separate market maker per outcome, ignoring all correlations). The underlying trader model for these results is a straightforward generalization of the Gaussian scoring model, and the details are deferred to Appendix~\prettyref{app:multinomial}.

\subparagraph{Results.}
Figure~\ref{fig:multinomial-key-stats} reports the key per-round statistics
(mean, median, 95\% VaR, and 95\% CVaR of Loss\,/\,Step) across the three
configurations.
In every scenario the Ising AMM sits strictly between the True Oracle lower
bound and the Independent baseline: it captures a large fraction of the
correlation structure available to the oracle while outperforming the
disjoint-market-maker strategy on all four statistics.
The gap between Ising AMM and Independent is particularly pronounced in the
tail-risk metrics (95\% VaR and CVaR), indicating that joint modelling not
only reduces average loss but also curtails the worst-case exposure that arises
when correlated outcomes are priced independently.
These results demonstrate that the Ising parameterisation and the parlay-driven
gradient signal are not artefacts of the binary setting. The mechanism
extends naturally to multinomial markets and retains its loss advantage over
disjoint market makers even as the outcome space grows.

\begin{figure}[t]
  \centering
  \begin{subfigure}[t]{0.32\textwidth}
    \includegraphics[width=\linewidth]{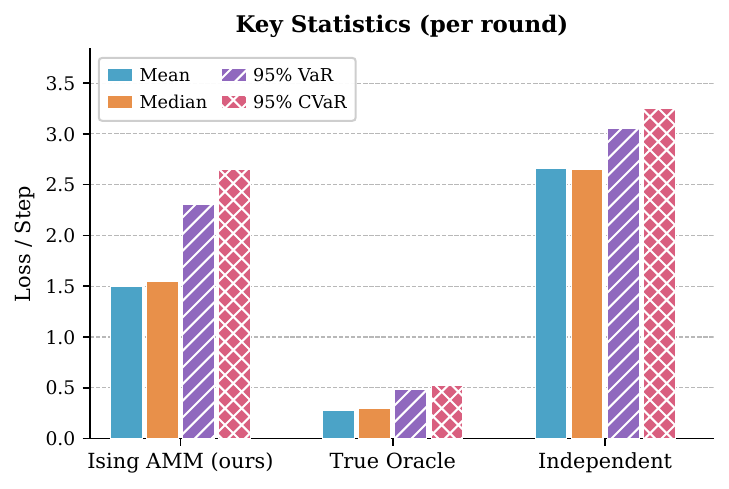}
    \caption{Five markets with $[2,3,4,4,5]$ outcomes}
    \label{fig:multinomial-1}
  \end{subfigure}
  \hfill
  \begin{subfigure}[t]{0.32\textwidth}
    \includegraphics[width=\linewidth]{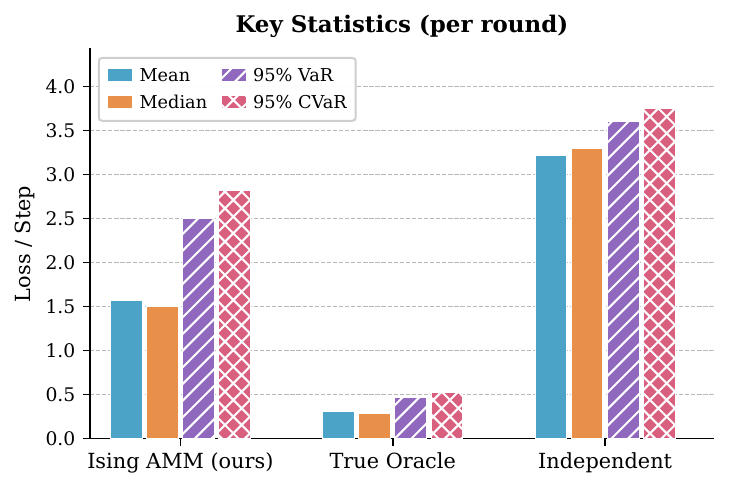}
    \caption{Eight markets with $[2,3,3,4,4,5,6,7]$ outcomes}
    \label{fig:multinomial-2}
  \end{subfigure}
  \hfill
  \begin{subfigure}[t]{0.32\textwidth}
    \includegraphics[width=\linewidth]{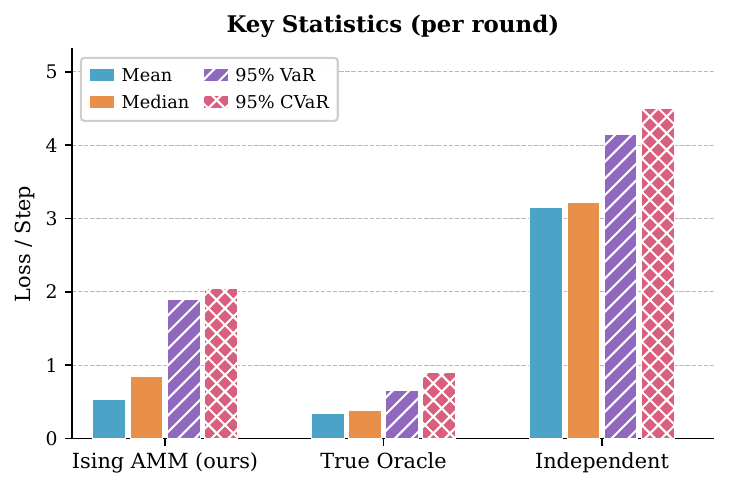}
    \caption{Eight markets with $[2,3,3,4,4,5,6,7]$ outcomes, and random jumps}
    \label{fig:multinomial-3}
  \end{subfigure}
  \caption{Key per-round statistics (Loss\,/\,Step) for three multinomial market
           configurations. In all cases the Ising AMM (ours) sits between the
           True Oracle lower bound and the Independent baseline, with the largest
           relative gain visible in the tail-risk measures (95\% VaR and CVaR).}
  \label{fig:multinomial-key-stats}
\end{figure}

\subsection{Empirical Market Performance}

To evaluate whether these advantages survive outside the synthetic setting, we replay the mechanism on historical Kalshi data. 
The empirical study focuses on three implementation details that are central to the practicality of the design.

\subparagraph{Two execution settings.} We consider \emph{two execution settings}. In a \emph{full-market emulation}, base markets are modeled as LMSRs under the AMM, and the informative signal extracted from each trade is its \emph{target price}: an informed trader is assumed to move the market to the price implied by her information.  In a \emph{standalone market-maker emulation}, base markets remain external and the AMM is used only as a shared quoting and underwriting engine for parlays; in that case the relevant signal is \emph{order flow} - trade direction and share quantity - rather than the final market price. These two settings correspond to two distinct practical roles of the mechanism: a complete market design and a drop-in correlation-aware parlay market maker layered on top of an existing RFQ-style platform.

\subparagraph{Real-market emulation.} The replay is performed on a \emph{real historical universe}. We use Kalshi's NBA slate for 2026-03-07, comprising 732 base markets across six games and 5,918 listed combo tickers with legs entirely inside this universe, spanning 2- through 10-leg parlays. To scale inference to clusters of this size, we train one Ising model per game cluster and use \emph{loopy belief propagation} rather than exact enumeration. Each model is initialized from the first 30 candle mid-prices of its constituent legs, with fields initialized from median log-odds and couplings initialized at zero, after which the stream of candle updates and parlay trades is replayed chronologically. Candle updates train local fields, while parlay trades update both fields and interactions through the trade-implied target price. The replay therefore evaluates performance on historically accepted RFQ
flow rather than on the full distribution of submitted RFQ requests.

\subparagraph{Preprocessing.} The empirical setup explicitly addresses a key market imperfection: \emph{fee bias}. Observed Kalshi combo execution prices are ask quotes rather than direct probability assessments, so they are systematically upward biased by platform markup. Training directly on these prices would therefore bias the learned joint model upward. To correct for this, we preprocess each observed trade into a canonical all-YES representation using the complement rule and inclusion--exclusion; for mixed YES/NO parlays, the additive fee cancels in the alternating expansion, yielding an approximately unbiased target in the model's native parameterization. This correction is essential in practice: without it, parlay loss remains flatter and revenue is lower.

The full replay pipeline and all implementation details are described in Appendix~\ref{app:historical_replay}. Table~\ref{tab:pool_sim} shows that standalone ParlayMarket market maker achieves the strongest risk-adjusted liquidity-provider portfolio, with the highest Sharpe ratio among the non-native strategies while still accepting a large fraction of historical flow. 
The full AMM accepts all trades but exhibits lower ROI, likely because the observed RFQ flow contains substantial noise that a fully internalized LMSR-style implementation must absorb directly. Taken together, these results show that the empirical value of ParlayMarket pricing is strongest in the \emph{standalone market-maker} regime: the model can be inserted into an existing platform, improve LP portfolio quality, and remain competitive on accepted quotes.

\subsection{Summary}
Across both synthetic and historical evaluations, the results support the same conclusion. 
A shared, learned representation of the joint distribution allows ParlayMarket to quote a combinatorially large family of contracts while controlling aggregate loss, using parlay order flow as a direct signal for dependence learning. 
In synthetic settings this produces the predicted scaling and convergence behavior; in historical replay it yields a practically useful correlation-aware quoting engine with strong risk-adjusted performance that can be readily deployed.

\begin{table}[t]
\centering
\caption{Emulation on the March 7, 2026 NBA slate under the market-competitive filter.}
\label{tab:pool_sim}
\begin{tabular}{lccrrr}
\toprule
Strategy & Accepted Trades & MM Win Rate & Net PnL & Sharpe \\
\midrule
Kalshi (execution price)         & 5436 & 84.5\% & \$244{,}190.77 &  0.7151 \\
Standalone MM (\(\theta_{\rm corr}\))    & 3973 & 86.8\%  & \$113{,}984.58 &  1.0049 \\
Full AMM                             & 5436 & 69.0\% & \$158{,}577.66       & 0.7062 \\
\bottomrule
\end{tabular}
\end{table}




\section{Conclusion and limitations}
\label{sec:conclusion}

In this paper, we showed that parlay trading can be incorporated into an automated market maker without requiring a separate pricing and risk-management process for every combination. The mechanism supports \(2^M\) base and joint contracts using only an \(O(M^2)\) shared representation, and correspondingly achieves \(O(M^2)\) aggregate information-retrieval loss rather than loss that scales with the full combinatorial market space. This shared structure also allows parlay order flow to serve as a direct signal for learning dependence structure. The empirical results further suggest that the mechanism is not only theoretically well-founded but close to deployable in practice: historical-market replay on Kalshi data shows that, in the standalone market-maker setting, correlation-aware quoting remains competitive under observed prices, RFQ flow, and realistic execution constraints.

An important scope condition of the current framework is that the traded contracts are assumed to already be expressed in a structurally coherent binary form. In many practical settings, this representation must itself be constructed before correlation-aware pricing becomes meaningful. For example, families of mutually exclusive outcomes may satisfy constraints such as $\sum_{i=1}^k X_i = 1$, even though the underlying latent process remains stochastic. In such cases, a separate preprocessing or structural-learning layer is needed to map raw market data into variables on which the assumptions of the AMM hold. We illustrate one such preprocessing step in the evaluation section. This suggests a broader architecture for automated parlay markets: a structural layer that identifies and normalizes the relevant event representation, followed by a pricing layer such as ParlayMarket that provides automated execution and, as a secondary benefit, elicits dependence from the resulting trade flow.


\bibliographystyle{plainurl}
\bibliography{references}


\newpage
\appendix

\section{Additional material for \prettyref{sec:evaluation}}\label{app:plots}

\begin{table}[t]
  \centering
  \caption{Per-market loss decay in the complete-parlay regime.
           $N_M = 2^M - 1$; values averaged over 100 simulations.}
  \label{tab:per-market-ratios}
  \begin{tabular}{cccccc}
    \hline
    $M$ & $N_M$ & $\ell_M$ & $\ell_{M+1}/\ell_M$ & $\ell_M \cdot 2^M$ & $N_M \cdot \ell_M$ \\
    \hline
    4 & 15  & $\sim 0.0094$ & $\sim 0.49$ & $\sim 0.15$ & $\sim 0.14$ \\
    5 & 31  & $\sim 0.0048$ & $\sim 0.50$ & $\sim 0.15$ & $\sim 0.15$ \\
    6 & 63  & $\sim 0.0024$ & $\sim 0.50$ & $\sim 0.15$ & $\sim 0.15$ \\
    7 & 127 & $\sim 0.0012$ & $\sim 0.50$ & $\sim 0.15$ & $\sim 0.15$ \\
    8 & 255 & $\sim 0.0006$ & $\sim 0.50$ & $\sim 0.15$ & $\sim 0.15$ \\
    9 & 511 & $\sim 0.0003$ & ---         & $\sim 0.15$ & $\sim 0.15$ \\
    \hline
  \end{tabular}
\end{table}

\begin{figure}[t]
  \centering
  \includegraphics[width=0.7\columnwidth]{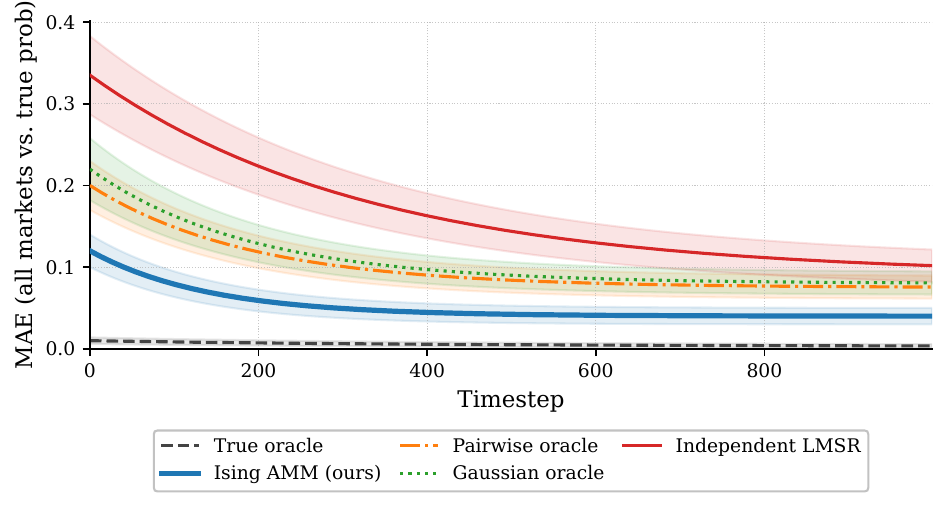}
  \caption{Price convergence over time (mean $\pm 1\sigma$ across 1000 simulations). LMSR market-price MAE vs.\ true probabilities.}
  \label{fig:price-convergence}
\end{figure}

\begin{figure}[t]
  \centering
  \begin{subfigure}[t]{0.48\columnwidth}
    \includegraphics[width=\columnwidth]{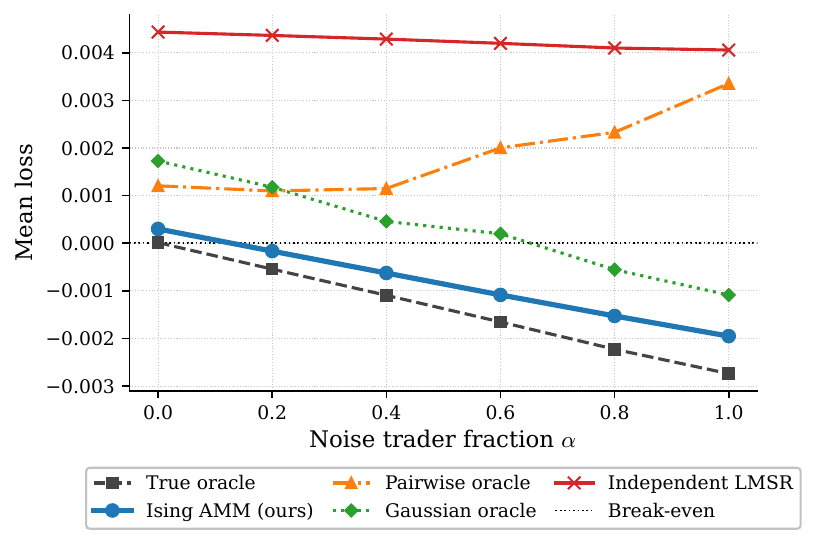}
    \caption{Per-market loss vs.\ noise trader fraction $\alpha$ ($M = 9$).}
    \label{fig:noise-loss}
  \end{subfigure}
  \hfill
  \begin{subfigure}[t]{0.48\columnwidth}
    \includegraphics[width=\columnwidth]{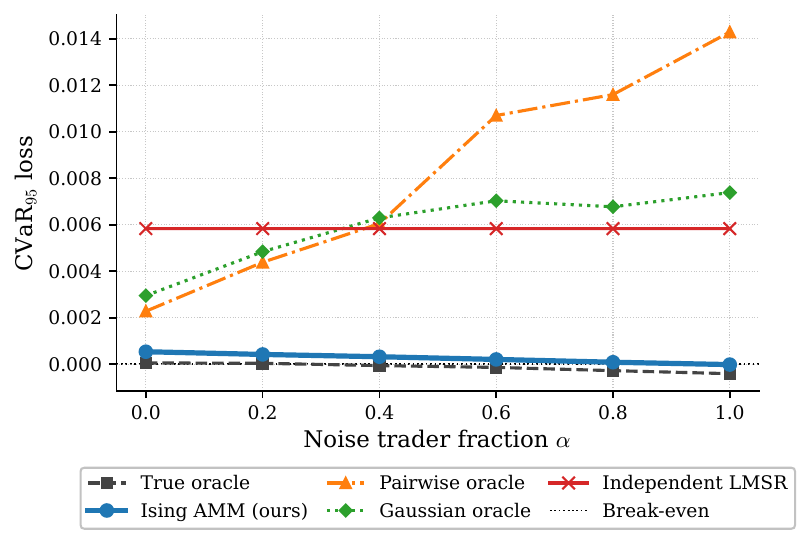}
    \caption{$\mathrm{CVaR}_{95}$ loss vs.\ noise trader fraction $\alpha$ ($M = 9$).}
    \label{fig:noise-cvar}
  \end{subfigure}
  \caption{Per-market loss under varying noise-trader fractions ($M = 9$,
           extrapolated from $M \in \{3,4,5\}$ simulations).
           The ParlayMarket remains the best-performing non-oracle model
           across all noise levels; the break-even line is crossed only at
           high noise fractions.}
\end{figure}
\section{RFQ Details}\label{rfq_details}
\subsection{Combo Mispricing for Independent Events}\label{app:rfq-misprice}
Our experiments show that the RFQ combo market fails Property~3 in a basic sense: combo prices do not reliably encode joint probabilities even in independent cases. We consider two-leg combos whose legs are drawn from different leagues (NBA and ATP tennis) so that the events are naturally close to independent. We fetch the base market prices for legs and compare them to the best quote we receive from the combo market for the 2-leg combo. Detailed quotes and prices for this experiment are shown in Table~\ref{tab:independent-combo-mispricing}. 

Under probabilistic pricing, the combo YES price should be approximately the product of the two leg YES prices. All 33 quoted prices exceed the independence benchmark. Moreover, and more decisively, 13 of the 33 quotes (39\%) violate the Fr\'{e}chet upper bound $P(A \cap B) \leq \min(P(A), P(B))$---a constraint that holds under \emph{any} joint probability distribution, regardless of correlation structure (Table~\ref{tab:frechet-violations}).

\begin{table}[t]                                                                                                                 
  \centering                                                                                                                       
  \caption{Quoted combo YES prices versus probabilistic benchmark for two-leg combos with plausibly independent legs (NBA and ATP
  tennis).}                                                                                                                        
  \label{tab:independent-combo-mispricing}
  \begin{tabular}{ccccc}                                                                                                           
  \toprule        
  Leg\,1 & Leg\,2 & $p_1 \cdot p_2$ & Quoted & $\Delta$ \\
  \midrule
  72c & 88c & 63.4c & 72c & 8.6c \\
  72c & 37c & 26.6c & 70c & 43.4c \\
  72c & 15c & 10.8c & 31c & 20.2c \\
  72c & 14c & 10.1c & 30c & 19.9c \\
  72c & 41c & 29.5c & 50c & 20.5c \\
  72c & 20c & 14.4c & 34c & 19.6c \\
  72c & 82c & 59.0c & 63c & 4.0c \\
  72c & 79c & 56.9c & 62c & 5.1c \\
  72c & 17c & 12.2c & 14c & 1.8c \\
  72c & 42c & 30.2c & 31c & 0.8c \\
  72c & 30c & 21.6c & 25c & 3.4c \\
  
  29c & 88c & 25.5c & 29c & 3.5c \\
  29c & 37c & 10.7c & 29c & 18.3c \\
  29c & 15c & 4.4c & 29c & 24.7c \\
  29c & 14c & 4.1c & 24c & 19.9c \\
  29c & 41c & 11.9c & 29c & 17.1c \\
  29c & 20c & 5.8c & 26c & 20.2c \\
  29c & 82c & 23.8c & 27c & 3.2c \\
  29c & 79c & 22.9c & 26c & 3.1c \\
  29c & 17c & 4.9c & 6c & 1.1c \\
  29c & 42c & 12.2c & 13c & 0.8c \\
  29c & 30c & 8.7c & 29c & 20.3c \\
  
  82c & 88c & 72.2c & 82c & 9.8c \\
  82c & 37c & 30.3c & 79c & 48.7c \\
  82c & 15c & 12.3c & 80c & 67.7c \\
  82c & 14c & 11.5c & 80c & 68.5c \\
  82c & 41c & 33.6c & 35c & 1.4c \\
  82c & 20c & 16.4c & 80c & 63.6c \\
  82c & 82c & 67.2c & 72c & 4.8c \\
  82c & 79c & 64.8c & 70c & 5.2c \\
  82c & 17c & 13.1c & 14c & 0.9c \\
  82c & 42c & 34.4c & 35c & 0.6c \\
  82c & 30c & 24.6c & 45c & 20.4c \\
  \bottomrule
  \end{tabular}
\end{table}

\begin{table}[t]
  \centering
  \caption{Combo quotes violating the Fr\'{e}chet upper bound $P(A \cap B) \leq \min(P(A), P(B))$, sorted by severity. These prices
   are incompatible with \emph{any} joint probability distribution, regardless of correlation.}
  \label{tab:frechet-violations}
  \begin{tabular}{cccccc}
  \toprule
  Leg\,1 & Leg\,2 & $p_1 \cdot p_2$ & $\min(p_1,p_2)$ & Quoted & $Q / \min$ \\
  \midrule
  82c & 14c & 11.5c & 14c & 80c & $5.71\times$ \\
  82c & 15c & 12.3c & 15c & 80c & $5.33\times$ \\
  82c & 20c & 16.4c & 20c & 80c & $4.00\times$ \\
  82c & 37c & 30.3c & 37c & 79c & $2.14\times$ \\
  72c & 14c & 10.1c & 14c & 30c & $2.14\times$ \\
  72c & 15c & 10.8c & 15c & 31c & $2.07\times$ \\
  29c & 15c & 4.4c & 15c & 29c & $1.93\times$ \\
  72c & 37c & 26.6c & 37c & 70c & $1.89\times$ \\
  82c & 30c & 24.6c & 30c & 45c & $1.50\times$ \\
  72c & 20c & 14.4c & 20c & 34c & $1.70\times$ \\
  29c & 14c & 4.1c & 14c & 24c & $1.71\times$ \\
  72c & 41c & 29.5c & 41c & 50c & $1.22\times$ \\
  29c & 20c & 5.8c & 20c & 26c & $1.30\times$ \\
  \bottomrule
  \end{tabular}
\end{table}

\subsection{Quote Aggregation to Tighten Spread}
\label{app:rfq-quote-agg}
We test whether multiple market makers materially tighten combo quotes. Let \(S1\) denote the best single-maker spread and \(S2\) the mixed spread obtained by taking the best YES quote and best NO quote across all market makers. From Table~\ref{tab:combo-quote-quality} the average \(S1\) is 12.08c and the average \(S2\) is 11.89c, so quote aggregation improves execution by only 0.19c on average. This pattern also holds across leg counts shown on Table~\ref{tab:combo-leg-count}: for 2-leg combos, average spreads remain very wide (26.02c for \(S1\), 24.30c for \(S2\)); for 3+ leg combos, the improvement is negligible. Overall, the presence of multiple market makers does not generate meaningful competitive tightening in combo markets.

\begin{table}[t]
\centering
\caption{Combo-market quote quality on Kalshi for markets with open interest between 1{,}000 and 2{,}000. The gain from split quoting is negligible.}
\label{tab:combo-quote-quality}
\begin{tabular}{lc}
\toprule
Metric & Value \\
\midrule
Average quotes per combo & 4.5 \\
Best single quote spread (\(S1\)) & 12.08c \\
Best mixed quote spread (\(S2\)) & 11.89c \\
Mixed-quote improvement (\(S1 - S2\)) & 0.19c \\
\bottomrule
\end{tabular}
\end{table}

\begin{table}[t]
\centering
\caption{Combo-market spreads by leg count on Kalshi.}
\label{tab:combo-leg-count}
\begin{tabular}{lccc}
\toprule
Leg count & Count & Avg.\ \(S1\) (cents) & Avg.\ \(S2\) (cents) \\
\midrule
2   & 43  & 26.02 & 24.30 \\
3   & 58  & 22.19 & 22.19 \\
4   & 56  & 14.80 & 14.73 \\
5+  & 392 & 8.66  & 8.60 \\
\bottomrule
\end{tabular}
\end{table}

\section{Mathematical Proofs}
\label{app:parlaymarket-proofs}

This appendix supplies proofs and supporting derivations for the claims stated in
Section~\ref{sec:correlation-mm}. In particular, it supports the Property~1 bound in
Subsection~\ref{subsec:amm-architecture}, Theorem~\ref{thm:linear-conv},
Proposition~\ref{prop:strong-convex}, Lemma~\ref{lem:quad-sensitivity},
Proposition~\ref{prop:per-market-decay}, and Proposition~\ref{prop:parlay-acceleration}.
Throughout this appendix we write $\phi$ for the parameter vector
$\boldsymbol{\varphi}$ from Section~\ref{sec:correlation-mm}.

\subsection{Proof supporting Property~1 in Subsection~\ref{subsec:amm-architecture}}\label{app:proof-prop-1}

\begin{proof}[Proof of the liquidity-consistency bound in the Property~1 paragraph of Subsection~\ref{subsec:amm-architecture}]
Fix any market $S$ and let $z_S$ denote the YES-minus-NO share imbalance in the
corresponding two-outcome LMSR book. With common liquidity parameter $b$, the LMSR YES
price is
\[
  p_S(z_S)
  =
  \frac{e^{z_S/b}}{1+e^{z_S/b}}.
\]
After a small YES order of size $x$, the new imbalance is $z_S+x$, so the local price impact is
\[
  I_S(x;\phi)
  :=
  p_S(z_S+x)-p_S(z_S).
\]
Since
\[
  p_S'(z_S)
  =
  \frac{1}{b}\,p_S(z_S)\bigl(1-p_S(z_S)\bigr),
\]
Taylor expansion gives
\[
  I_S(x;\phi)
  =
  \frac{x}{b}\,p_S^{\phi}\bigl(1-p_S^{\phi}\bigr)+O(x^2),
\]
which is the local-impact formula displayed in Subsection~\ref{subsec:amm-architecture}.

Assume now that all active quotes lie in the operating region
\[
  p_S^{\phi}\in[\varepsilon,1-\varepsilon]
  \qquad \text{for every active }S.
\]
Then
\[
  \varepsilon(1-\varepsilon)
  \le
  p_S^{\phi}\bigl(1-p_S^{\phi}\bigr)
  \le
  \frac14.
\]
Hence the directional derivative of price impact at the origin satisfies
\[
  \frac{\partial_x I_S(0;\phi)}{\partial_x I_{S'}(0;\phi)}
  =
  \frac{p_S^{\phi}(1-p_S^{\phi})}{p_{S'}^{\phi}(1-p_{S'}^{\phi})}
  \in
  \bigl[4\varepsilon(1-\varepsilon),\,1/(4\varepsilon(1-\varepsilon))\bigr].
\]
By continuity, the same comparability holds for all sufficiently small admissible trades $x$.
Thus small orders induce comparable local price impact across base and parlay contracts. Because
all quoted probabilities are functions of the same shared state $\phi$, every real trade also shifts
all related prices immediately, which is exactly the liquidity-consistency claim made in the
Property~1 paragraph of Subsection~\ref{subsec:amm-architecture}.
\end{proof}

\subsection{Proofs for Subsection~\ref{subsec:convergence}}

\subparagraph{Common notation.}
Let
\[
\mathcal E_2
:=
\bigl\{\{i\}:1\le i\le M\bigr\}
\cup
\bigl\{\{i,j\}:1\le i<j\le M\bigr\}.
\]
For a traded event $E\subseteq [M]$, write
\[
q_E(\phi)
:=
P_\phi\!\Big(\bigcap_{i\in E}\{X_i=1\}\Big),
\qquad
p_E^\ast
:=
P^\ast\!\Big(\bigcap_{i\in E}\{X_i=1\}\Big),
\]
and
\[
\ell_E(\phi):=\mathrm{CE}(p_E^\ast,q_E(\phi))
=
-p_E^\ast \log q_E(\phi) - (1-p_E^\ast)\log(1-q_E(\phi)).
\]
If the market corresponding to $E$ is sampled with probability $\lambda_E>0$, define
\[
L(\phi):=\sum_{E\in\mathcal E_2}\lambda_E\,\ell_E(\phi).
\]
This is the same singleton-and-pairwise composite objective as $\mathcal{L}$ in
\eqref{eq:pseudo-loss}, written with a shorter symbol in the appendix. Let
\[
T(X):=\bigl(X_1,\dots,X_M,\,(X_iX_j)_{1\le i<j\le M}\bigr)
\]
denote the vector of sufficient statistics. We assume throughout that the target moment vector
\[
m^\ast := \bigl(p_i^\ast,\; p_{ij}^\ast\bigr)_{1\le i<j\le M}
\]
lies in the interior of the pairwise marginal polytope, so that it is realizable by a unique
pairwise Ising parameter $\phi^\star$.

\subparagraph{Expected gradient of the sampled update.}
The sampled update in \eqref{eq:sgd-update} is exactly the stochastic-gradient step for the
composite loss \eqref{eq:pseudo-loss}. Indeed, for any traded event $E$,
\[
\nabla_\phi \ell_E(\phi)
=
\frac{q_E(\phi)-p_E^\ast}{q_E(\phi)(1-q_E(\phi))}\,\nabla_\phi q_E(\phi).
\]
By the covariance identity used in Subsection~\ref{subsec:convergence},
\[
\nabla_\phi q_E(\phi)
=
\operatorname{Cov}_\phi\!\bigl(T(X),\mathbf 1_E\bigr)
=
q_E(\phi)\Bigl(\mathbb E_\phi[T(X)\mid E]-\mathbb E_\phi[T(X)]\Bigr).
\]
Hence
\[
\nabla_\phi \ell_E(\phi)
=
\frac{q_E(\phi)-p_E^\ast}{1-q_E(\phi)}
\Bigl(\mathbb E_\phi[T(X)\mid E]-\mathbb E_\phi[T(X)]\Bigr),
\]
which is exactly the vector form of the marketwise update described after \eqref{eq:sgd-update}.
Therefore, if $E_t$ denotes the event traded at round $t$, then
\[
\phi_{t+1}=\phi_t-\eta g_t,
\qquad
\mathbb E[g_t\mid \phi_t]=\nabla L(\phi_t).
\]

\begin{proof}[Proof of the moment-matched target claim in Subsection~\ref{subsec:convergence}]
Because the mean dynamics are gradient descent on $L$, it suffices to identify the unique
minimizer of $L$. Let $\phi^\star$ be the unique Ising parameter whose singleton and pairwise
moments equal $(p_i^\ast,p_{ij}^\ast)$. Such a parameter exists and is unique because the
pairwise Ising family is a finite exponential family with strictly convex log-partition function,
so the map from natural parameters to interior moments is one-to-one.

For every $E\in\mathcal E_2$,
\[
\ell_E(\phi)-\ell_E(\phi^\star)
=
\mathrm{CE}(p_E^\ast,q_E(\phi))-\mathrm{CE}(p_E^\ast,p_E^\ast)
=
\mathrm{KL}\!\bigl(\mathrm{Bern}(p_E^\ast)\,\|\,\mathrm{Bern}(q_E(\phi))\bigr)
\ge 0.
\]
Summing with weights $\lambda_E$ gives
\[
L(\phi)-L(\phi^\star)
=
\sum_{E\in\mathcal E_2}\lambda_E\,
\mathrm{KL}\!\bigl(\mathrm{Bern}(p_E^\ast)\,\|\,\mathrm{Bern}(q_E(\phi))\bigr)
\ge 0.
\]
Equality holds if and only if $q_E(\phi)=p_E^\ast$ for every traded event $E$, that is, if and
only if $\phi$ matches all singleton and pairwise moments. By injectivity of the Ising moment
map on the interior of the marginal polytope, this happens only at $\phi=\phi^\star$.

Thus $\phi^\star$ is the unique global minimizer of $L$. Since the mean update is
$-\eta \nabla L(\phi)$, $\phi^\star$ is also the unique fixed point of the mean dynamics. Finally,
the maximum-entropy characterization in the moment-matched target paragraph follows from the
standard Lagrange-duality argument for finite exponential families: among all distributions on
$\{0,1\}^M$ with the prescribed first and second moments, the entropy maximizer has density
proportional to $\exp(\langle \phi^\star,T(x)\rangle)$, which is exactly the pairwise Ising law.
\end{proof}

\begin{proof}[Proof of Proposition~\ref{prop:strong-convex} (local strong-convexity form used in the convergence argument)]
Let
\[
m(\phi):=\bigl(q_E(\phi)\bigr)_{E\in\mathcal E_2}\in\mathbb R^d,
\qquad
J(\phi):=\nabla_\phi m(\phi)\in\mathbb R^{d\times d}.
\]
Write $L=F\circ m$, where
\[
F(u):=\sum_{E\in\mathcal E_2}\lambda_E\,\mathrm{CE}(p_E^\ast,u_E).
\]
Since
\[
\frac{\partial}{\partial u_E}\mathrm{CE}(p_E^\ast,u_E)
=
\frac{u_E-p_E^\ast}{u_E(1-u_E)},
\qquad
\frac{\partial^2}{\partial u_E^2}\mathrm{CE}(p_E^\ast,u_E)
=
\frac{1}{u_E(1-u_E)},
\]
the chain rule gives
\[
\nabla^2 L(\phi)
=
J(\phi)^\top D(\phi)J(\phi)
+
\sum_{E\in\mathcal E_2}\lambda_E
\frac{q_E(\phi)-p_E^\ast}{q_E(\phi)(1-q_E(\phi))}
\,\nabla_\phi^2 q_E(\phi),
\]
where
\[
D(\phi)
=
\operatorname{diag}\!\Bigl(\frac{\lambda_E}{q_E(\phi)(1-q_E(\phi))}\Bigr)_{E\in\mathcal E_2}.
\]
At $\phi=\phi^\star$ we have $q_E(\phi^\star)=p_E^\ast$ for every $E$, so the second term
vanishes and therefore
\[
\nabla^2 L(\phi^\star)
=
J(\phi^\star)^\top D^\star J(\phi^\star),
\qquad
D^\star
=
\operatorname{diag}\!\Bigl(\frac{\lambda_E}{p_E^\ast(1-p_E^\ast)}\Bigr)_{E\in\mathcal E_2}.
\]
Assume now that $J(\phi^\star)$ has full rank and that $p_E^\ast\in[\varepsilon,1-\varepsilon]$
for every traded event $E$. Then $D^\star\succ 0$, hence $\nabla^2 L(\phi^\star)\succ 0$. Let
\[
\mu_0:=\lambda_{\min}\!\bigl(\nabla^2 L(\phi^\star)\bigr)>0.
\]
Because $\nabla^2L$ is continuous, there exists a neighborhood $U$ of $\phi^\star$ such that
\[
\nabla^2L(\phi)\succeq \frac{\mu_0}{2}I
\qquad \text{for all }\phi\in U.
\]
By the Hessian characterization of strong convexity, $L$ is $\mu$-strongly convex on $U$ with
$\mu=\mu_0/2$.

This is the local strong-convexity statement needed in the proof of
Theorem~\ref{thm:linear-conv}. The explicit Hessian values quoted in
Proposition~\ref{prop:strong-convex} at the uninformative initialization $\phi=0$ are obtained by
evaluating the same Fisher blocks under the uniform product measure.
\end{proof}

\begin{proof}[Proof of Theorem~\ref{thm:linear-conv}]
Let $U$ be the neighborhood from the previous proof, and assume that all iterates remain in
$U$. Write
\[
e_t:=\phi_t-\phi^\star.
\]
Because the sampled gradient is unbiased, we may write
\[
g_t=\nabla L(\phi_t)+\zeta_t,
\qquad
\mathbb E[\zeta_t\mid \phi_t]=0,
\qquad
\mathbb E[\|\zeta_t\|^2\mid \phi_t]\le \sigma^2.
\]
Let $L_g$ denote a gradient-Lipschitz constant of $L$ on $U$. Then
\[
e_{t+1}=e_t-\eta \nabla L(\phi_t)-\eta \zeta_t.
\]
Taking conditional expectations and using $\mathbb E[\zeta_t\mid \phi_t]=0$,
\[
\mathbb E\!\left[\|e_{t+1}\|^2\mid \phi_t\right]
=
\|e_t\|^2
-2\eta\langle \nabla L(\phi_t),e_t\rangle
+\eta^2\|\nabla L(\phi_t)\|^2
+\eta^2\mathbb E\!\left[\|\zeta_t\|^2\mid \phi_t\right].
\]
Local strong convexity gives
\[
\langle \nabla L(\phi_t),e_t\rangle \ge \mu \|e_t\|^2,
\]
and $L_g$-smoothness gives
\[
\|\nabla L(\phi_t)\|
=
\|\nabla L(\phi_t)-\nabla L(\phi^\star)\|
\le L_g \|e_t\|.
\]
Hence
\[
\mathbb E\!\left[\|e_{t+1}\|^2\mid \phi_t\right]
\le
\bigl(1-2\eta\mu+\eta^2 L_g^2\bigr)\|e_t\|^2+\eta^2\sigma^2.
\]
Set
\[
\rho:=1-2\eta\mu+\eta^2L_g^2.
\]
Taking full expectations and iterating the recursion yields
\[
\mathbb E\|e_t\|^2
\le
\rho^t \|e_0\|^2
+
\eta^2\sigma^2\sum_{s=0}^{t-1}\rho^s
\le
\rho^t \|e_0\|^2
+
\frac{\eta^2\sigma^2}{1-\rho}
=
\rho^t \|e_0\|^2
+
\frac{\eta\sigma^2}{2\mu-\eta L_g^2}.
\]
This is the precise finite-$\eta$ estimate underlying the display in
Theorem~\ref{thm:linear-conv}.

If $\eta\le \mu/L_g^2$, then $\rho\le 1-\eta\mu$, so
\[
\mathbb E\|e_t\|^2
\le
(1-\eta\mu)^t\|e_0\|^2+\frac{2\eta\sigma^2}{\mu}.
\]
In the small-step regime,
\[
\rho=1-2\eta\mu+O(\eta^2),
\qquad
\frac{\eta\sigma^2}{2\mu-\eta L_g^2}
=
\frac{\eta\sigma^2}{2\mu}+O(\eta^2),
\]
which is exactly the simplified form quoted in Subsection~\ref{subsec:convergence}.
\end{proof}

\begin{proof}[Proof of the $O(\log T)$ cumulative-regret statement following \eqref{eq:mu-eff}]
Let $L_t(\phi):=\ell_{E_t}(\phi)$ be the random round-$t$ loss, where $E_t$ is the market
sampled at round $t$. Because $E_t$ is sampled with probabilities $(\lambda_E)_{E\in\mathcal E_2}$,
\[
\mathbb E[L_t(\phi_t)-L_t(\phi^\star)\mid \phi_t]
=
L(\phi_t)-L(\phi^\star).
\]
Now take the decaying step size $\eta_t=c/t$. From the same recursion as in the proof of
Theorem~\ref{thm:linear-conv},
\[
\mathbb E\|\phi_t-\phi^\star\|^2\le \frac{C}{t}
\]
for some constant $C>0$. Since $L$ is $L_g$-smooth and $\nabla L(\phi^\star)=0$,
\[
L(\phi_t)-L(\phi^\star)
\le \frac{L_g}{2}\|\phi_t-\phi^\star\|^2.
\]
Therefore
\[
\mathbb E\!\left[\sum_{t=1}^T \bigl(L_t(\phi_t)-L_t(\phi^\star)\bigr)\right]
=
\sum_{t=1}^T \mathbb E\bigl[L(\phi_t)-L(\phi^\star)\bigr]
\le
\frac{L_g C}{2}\sum_{t=1}^T \frac{1}{t}
\le
\frac{L_g C}{2}(1+\log T).
\]
Thus the expected cumulative excess loss is $O(\log T)$, exactly as stated after
\eqref{eq:mu-eff}.
\end{proof}

\subsection{Proofs for Subsection~\ref{subsec:per-market-decay}}

\begin{proof}[Proof of Lemma~\ref{lem:quad-sensitivity}]
Fix a market $S$ and define
\[
 f(p):=\mathrm{KL}(p_S^\ast\|p)
 =
 p_S^\ast\log\frac{p_S^\ast}{p}
 +(1-p_S^\ast)\log\frac{1-p_S^\ast}{1-p}.
\]
Then $f(p_S^\ast)=0$ and $f'(p_S^\ast)=0$, while
\[
 f''(p)
 =
 \frac{p_S^\ast}{p^2}+\frac{1-p_S^\ast}{(1-p)^2}.
\]
In particular,
\[
 f''(p_S^\ast)=\frac{1}{p_S^\ast(1-p_S^\ast)}.
\]
By Taylor's theorem, for $p$ in a sufficiently small neighborhood of $p_S^\ast$,
\[
 \mathrm{KL}(p_S^\ast\|p)
 =
 \frac{(p-p_S^\ast)^2}{2p_S^\ast(1-p_S^\ast)} + O\!\bigl(|p-p_S^\ast|^3\bigr).
\]
Next apply the mean-value theorem to the map $\phi\mapsto p_S^\phi$. For $\phi$ in a small
neighborhood $U$ of $\phi^\star$,
\[
 |p_S^\phi-p_S^\ast|
 \le
 \sup_{\tilde\phi\in U}\|\nabla p_S^{\tilde\phi}\|\,\|\phi-\phi^\star\|.
\]
Substituting this into the quadratic expansion above gives
\[
 b\,\mathrm{KL}(p_S^\ast\|p_S^\phi)
 \le
 \frac{b}{2}
 \left(
 \sup_{\tilde\phi\in U}
 \frac{\|\nabla p_S^{\tilde\phi}\|^2}{p_S^\ast(1-p_S^\ast)}
 \right)
 \|\phi-\phi^\star\|^2
 + O\!\bigl(\|\phi-\phi^\star\|^3\bigr).
\]
After shrinking $U$ if necessary, the cubic remainder is dominated by the quadratic term, so the
local bound in Lemma~\ref{lem:quad-sensitivity} follows with
\[
 \Lambda_S
 :=
 \sup_{\tilde\phi\in U}
 \frac{\|\nabla p_S^{\tilde\phi}\|^2}{p_S^\ast(1-p_S^\ast)}.
\]
The value $\Lambda_S=1/4$ quoted in Subsection~\ref{subsec:per-market-decay} is the
corresponding baseline constant obtained by evaluating the same gradient expression at the
symmetric initialization $\phi=0$.
\end{proof}

\begin{proof}[Proof of Proposition~\ref{prop:per-market-decay} and of the mean-loss corollary stated afterward]
For the complete-parlay regime, let
\[
N_M:=2^M-1.
\]
Let $\ell_t$ denote the expected LMSR monetary loss incurred at round $t$ by the market that is
actually traded. Define the family-average quadratic sensitivity
\[
\bar\Lambda_M
:=
\frac{1}{N_M}
\sum_{\varnothing\neq S\subseteq[M]}
\sup_{\phi\in U}
\frac{\|\nabla p_S^\phi\|^2}{p_S^\ast(1-p_S^\ast)}.
\]
Assume the following three properties hold with constants independent of $M$:
\begin{align*}
\text{(A1)}\qquad
\mathbb E\|\phi_t-\phi^\star\|^2
&\le
C_0 \|\phi^\star\|^2 e^{-c\eta t}+C_1\eta,
\\
\text{(A2)}\qquad
\bar\Lambda_M
&\le C_2,
\\
\text{(A3)}\qquad
\|\phi^\star\|^2
&=O(M^2\rho^2).
\end{align*}
Condition (A1) is the transient-plus-noise-floor estimate inherited from
Theorem~\ref{thm:linear-conv}; (A2) is the uniform average sensitivity bound needed to convert
parameter error into monetary loss through Lemma~\ref{lem:quad-sensitivity}; and (A3) is the
weak-correlation scaling of the target interaction vector used in the text.

For the market $S_t$ traded at round $t$, Lemma~\ref{lem:quad-sensitivity} gives
\[
\mathbb E[\ell_t\mid S_t=S]
\le
\frac{b}{2}
\left(
\sup_{\phi\in U}
\frac{\|\nabla p_S^\phi\|^2}{p_S^\ast(1-p_S^\ast)}
\right)
\mathbb E\|\phi_t-\phi^\star\|^2.
\]
Averaging over the uniformly sampled market family therefore yields
\[
\mathbb E[\ell_t]
\le
\frac{b\bar\Lambda_M}{2}\,
\mathbb E\|\phi_t-\phi^\star\|^2
\le
\frac{b C_2}{2}\Bigl(C_0\|\phi^\star\|^2 e^{-c\eta t}+C_1\eta\Bigr).
\]

Summing the transient term over all $t\ge 0$ and using
$\sum_{t\ge 0}e^{-c\eta t}=O(1/\eta)$, we obtain
\[
L_{\mathrm{transient}}
\le
\sum_{t\ge 0}\frac{bC_0C_2}{2}\|\phi^\star\|^2 e^{-c\eta t}
=
O\!\left(\frac{b\|\phi^\star\|^2}{\eta}\right)
=
O\!\left(\frac{b\rho^2 M^2}{\eta}\right),
\]
which is the transient bound stated in Proposition~\ref{prop:per-market-decay}.

The steady-state per-round loss is obtained from the $C_1\eta$ term:
\[
L_{\mathrm{ss,round}}
=
O(b\eta).
\]
This is the second display in Proposition~\ref{prop:per-market-decay}.

Finally, the mean loss per listed market discussed at the end of
Subsection~\ref{subsec:per-market-decay} satisfies
\[
\bar\ell_M
=
\frac{1}{N_M T}
\sum_{t=1}^T \mathbb E[\ell_t]
\le
\frac{L_{\mathrm{transient}}}{N_M T}
+
\frac{L_{\mathrm{ss,round}}}{N_M}.
\]
Hence
\[
\bar\ell_M
=
O\!\left(\frac{b\rho^2M^2}{\eta N_M T}\right)
+
O\!\left(\frac{b\eta}{N_M}\right).
\]
For $T\to\infty$, the steady-state term dominates, so
\[
\bar\ell_M
=
O\!\left(\frac{b\eta}{N_M}\right)
=
O(b\eta\,2^{-M}).
\]
Since $N_M=2^M-1$,
\[
\frac{\bar\ell_{M+1}}{\bar\ell_M}
\longrightarrow
\frac{N_M}{N_{M+1}}
=
\frac{2^M-1}{2^{M+1}-1}
\longrightarrow \frac12.
\]
This proves the halving-law corollary stated after Proposition~\ref{prop:per-market-decay}.
\end{proof}

\subparagraph{Remark on shadow trades.}
The previous proof shows exactly where shadow trades enter the argument in
Subsection~\ref{subsec:per-market-decay}: they are part of the mechanism that makes a
uniform-in-$M$ bound such as (A1) plausible. Without shadow repricing, parameters are updated only
from direct order flow in their own markets, so the steady-state term in (A1) need not remain
$O(\eta)$ uniformly in $M$; once that uniform bound fails, the aggregate loss guarantee in
Proposition~\ref{prop:per-market-decay} can fail as well.

\subsection{Proofs for Subsection~\ref{subsec:parlay-acceleration}}

\begin{proof}[Proof of Proposition~\ref{prop:parlay-acceleration}]
Let
\[
\mathcal E_+
:=
\mathcal E_2
\cup
\{S\subseteq [M]: |S|\ge 3,\ \lambda_S>0\},
\]
and define the augmented composite loss
\[
L^+(\phi):=\sum_{E\in \mathcal E_+}\lambda_E\,\mathrm{CE}(p_E^\ast,q_E(\phi)).
\]
This is the appendix shorthand for the augmented objective $\mathcal L^+$ used in
Proposition~\ref{prop:parlay-acceleration}. Assume there exists a parameter $\phi^\dagger$ such that
\[
q_E(\phi^\dagger)=p_E^\ast
\qquad \text{for every }E\in \mathcal E_+,
\]
and that the augmented Jacobian $J_+(\phi^\dagger)$ has full rank.

Exactly as in the proof of Proposition~\ref{prop:strong-convex}, the Hessian at $\phi^\dagger$ is
\[
\nabla^2L^+(\phi^\dagger)
=
\sum_{E\in \mathcal E_+}
\frac{\lambda_E}{p_E^\ast(1-p_E^\ast)}
\,\nabla q_E(\phi^\dagger)\nabla q_E(\phi^\dagger)^\top.
\]
Split the sum into baseline events and higher-order events:
\[
\nabla^2L^+(\phi^\dagger)
=
\nabla^2L(\phi^\dagger)
+
\sum_{\substack{S\subseteq [M]\\ |S|\ge 3}}
\frac{\lambda_S}{p_S^\ast(1-p_S^\ast)}
\,\nabla q_S(\phi^\dagger)\nabla q_S(\phi^\dagger)^\top.
\]
Every added term is positive semidefinite, so
\[
\nabla^2L^+(\phi^\dagger)-\nabla^2L(\phi^\dagger)\succeq 0.
\]
Hence the local strong-convexity constant cannot decrease:
\[
\mu_{\mathrm{eff}}\ge \mu.
\]
Moreover, if one of the added rank-one terms has nonzero projection on a weakest-curvature
eigendirection of $\nabla^2L(\phi^\dagger)$, then $\mu_{\mathrm{eff}}>\mu$, which is exactly the
curvature improvement summarized in \eqref{eq:mu-eff} and in item~1 of
Proposition~\ref{prop:parlay-acceleration}.

Applying Theorem~\ref{thm:linear-conv} with $L^+$ in place of $L$ yields the improved contraction
factor and the smaller steady-state error floor in item~2. Averaging the resulting bound over
$t\le T$ gives
\[
\bar\varepsilon_T^2
:=
\frac1T\sum_{t=1}^T \mathbb E\|\phi_t-\phi^\dagger\|^2
\le
C\left(
\frac{\|\phi_0-\phi^\dagger\|^2}{\eta \mu_{\mathrm{eff}} T}
+\eta
\right)
\]
for a constant $C$ depending only on local smoothness and variance parameters. Therefore any
quadratic price--loss bound of the form
\[
\ell_{\mathrm{mon}}(t)\le C_b \|\phi_t-\phi^\dagger\|^2
\]
immediately implies the improved average-loss conclusion in item~3.
\end{proof}

\section{Multinomial Gaussian Scoring Model}
\label{app:multinomial}

This appendix describes the generalization of the binary Gaussian scoring model
(Section~3.2) to categorical markets used in Section~\ref{subsec:eval-multinomial}.

\subparagraph{Categorical markets.}
We consider $M$ base markets where market $i$ can resolve to one of $K_i \geq 2$
mutually exclusive outcomes.
The experiment in Section~\ref{subsec:eval-multinomial} uses two configurations:
$\mathbf{K} = [2,3,4,4,5]$ and $\mathbf{K} = [2,3,3,4,4,5,6,7]$,
giving total outcome spaces of $\prod_i K_i = 480$ and $241{,}920$ respectively.

\subparagraph{Score process and threshold model.}
Each market $i$ is driven by an underlying score $S_i(t) \in \mathbb{R}$ following
correlated arithmetic Brownian motion
\begin{equation}
  dS_i(t) = \sigma\, dW_i(t), \qquad
  \mathrm{Cov}(dW_i, dW_j) = \rho_{ij}\,dt,
  \label{eq:multi-score}
\end{equation}
with constant volatility $\sigma = 1$ and pairwise correlation
$\rho_{ij} = \rho = 0.3$ for $i \neq j$ (equicorrelation structure).

The outcome of market $i$ is determined by $K_i - 1$ ordered thresholds
$\tau_i^1 < \cdots < \tau_i^{K_i - 1}$:
\begin{equation}
  X_i = k \iff S_i(T) \in \bigl(\tau_i^{k-1},\, \tau_i^k\bigr],
  \label{eq:multi-outcome}
\end{equation}
where we set $\tau_i^0 = -\infty$ and $\tau_i^{K_i} = +\infty$.
Thresholds are initialized so that all $K_i$ categories are equiprobable at $t = 0$:
\begin{equation}
  \tau_i^k = \sigma\sqrt{T}\,\Phi^{-1}\!\left(\tfrac{k}{K_i}\right),
  \quad k = 1, \ldots, K_i - 1.
  \label{eq:multi-thresholds}
\end{equation}

\subparagraph{Informed trader signals.}
Given the current score $S_i(t) = s$ and time remaining $\tau = T - t$,
the conditional probability of each outcome is
\begin{equation}
  P(X_i = k \mid S_i(t) = s)
    = \Phi\!\left(\frac{\tau_i^k - s}{\sigma\sqrt{\tau}}\right)
    - \Phi\!\left(\frac{\tau_i^{k-1} - s}{\sigma\sqrt{\tau}}\right).
  \label{eq:multi-marginal}
\end{equation}
For joint events, bivariate and trivariate normal CDFs are evaluated
via inclusion--exclusion on the hyperrectangle defined by the relevant
threshold intervals.
Specifically, for a pair $(i, j)$,
\begin{equation}
  P(X_i = k,\, X_j = l \mid \mathbf{S}_t)
  = \sum_{a \in \{k-1, k\}} \sum_{b \in \{j-1, j\}}
    (-1)^{a+b}\,
    \Phi_2\!\left(\frac{\tau_i^a - s_i}{\sigma\sqrt{\tau}},\;
                   \frac{\tau_j^b - s_j}{\sigma\sqrt{\tau}};\;
                   \rho_{ij}\right),
  \label{eq:multi-joint}
\end{equation}
and analogously for 3-way subsets using the trivariate normal CDF $\Phi_3$.

\subparagraph{Categorical parlays.}
A categorical parlay is a pair $(S, \mathbf{c})$ where $S \subseteq [M]$ is a subset
of markets and $\mathbf{c} = (c_i)_{i \in S}$ is a target outcome vector.
The parlay pays if and only if $X_i = c_i$ for every $i \in S$.
Parlay books remain two-state LMSRs (YES/NO), with the true price given by
$P^*_{S,\mathbf{c}} = P(\bigcap_{i\in S}\{X_i = c_i\} \mid \mathbf{S}_t)$.
We generate all 2- and 3-market categorical parlays, capped at 300 per leg count.
Trade-level weights are set to $50\%$ base, $30\%$ 2-leg, and $20\%$ 3-leg.

\subparagraph{Market maker: Potts MRF.}
The binary Ising model generalizes naturally to the \emph{Potts MRF}:
\begin{equation}
  P_\varphi(\mathbf{x}) \propto \exp\!\left(
    \sum_{i=1}^M \theta_i[x_i]
    + \sum_{i < j} W_{ij}[x_i, x_j]
  \right),
  \label{eq:potts}
\end{equation}
where $\theta_i \in \mathbb{R}^{K_i}$ are bias vectors
and $W_{ij} \in \mathbb{R}^{K_i \times K_j}$ are pairwise interaction matrices.
Identifiability is enforced by fixing $\theta_i[0] = 0$ and
$W_{ij}[0, :] = W_{ij}[:, 0] = \mathbf{0}$ (first category as reference).
The parameter count scales as $O\!\left(\sum_i K_i + \sum_{i<j} K_i K_j\right)$,
which is $O(K^2 M^2)$ for uniform $K_i = K$---polynomial in $M$ and $K$.
The SGD update rule applies unchanged with $\nabla_\varphi$ taken
with respect to the Potts parameters.
For base markets we use a $K_i$-state LMSR with cost function
$C(\mathbf{s}) = b\log\sum_k \exp(s_k/b)$ and price vector $\mathrm{softmax}(\mathbf{s}/b)$.

\subparagraph{Jump variant.}
The third configuration in \prettyref{fig:multinomial-key-stats} additionally overlays a Poisson jump process
on the score paths: at each time step, each market $i$ independently draws a
Poisson count $N_i \sim \mathrm{Poi}(\lambda)$ with $\lambda = 0.01$, and on each
jump the score is relocated to one of market $i$'s finite thresholds chosen uniformly
at random.
This mimics discontinuous information arrivals (e.g., injury announcements or
score shocks) and tests robustness of the Potts MRF update to sudden distributional
shifts.

\subparagraph{Simulation parameters.}
All experiments use $T = 1$, $300$ time steps ($dt = 1/300$), $b = 10$,
$\eta_\theta = \eta_W = 0.2$, $200$ independent simulation runs,
and zero noise-trader fraction.
Exact Potts inference is used when the total state space satisfies
$\prod_i K_i \leq 4096$; loopy belief propagation is used otherwise.

\section{Historical-Market Replay Details}
\label{app:historical_replay}

This appendix summarizes the implementation details of the historical-market replay
used in Section~6.

\subsection{Execution Settings}

We study two execution settings. In the \emph{full-market emulation}, base markets are
modeled as externally operated LMSRs and parlay prices are built on top of those market
states. The signal extracted from each trade is its \emph{target price}: we assume that
an informed trader moves the market to the price implied by her information, and the AMM
updates from that post-trade price. In the \emph{standalone market-maker emulation},
base markets remain operated by the platform and parlay trades are executed directly
against a shared liquidity pool. In this case, the informative signal is the trader's
\emph{order flow} - trade direction and share quantity - rather than the final market
price.

\subsection{Data and Model Setup}

We run the historical-market emulation on Kalshi's NBA slate for 2026-03-07. Learning
rates are set to $\eta_\theta = 0.1$ and $\eta_W = 0.01$, and the LMSR liquidity
parameter is $b = 10{,}000$. On this slate, Kalshi lists 6 NBA games and 732 available
base markets across nine series, including moneyline, spread, total, and player-prop
contracts. These markets partition into 6 event-key clusters, one per game, with cluster
sizes ranging from 73 to 135 markets. We match 5,918 listed combo tickers whose legs lie
entirely in this universe, spanning 2- through 10-leg parlays, and retain executed RFQ
trades only. 

We train one Ising model per cluster. Because these clusters contain many markets,
inference is performed with loopy belief propagation rather than exact enumeration.
Each model is initialized from the first 30 candle mid-prices of its constituent legs, with
\[
\theta_i = \mathrm{logit}\!\bigl(\mathrm{median}_t\, p_{it}\bigr),
\qquad
W_{ij}=0.
\]
The merged stream of 1-minute candle updates and parlay trades is then replayed in
chronological order. Candle updates train the local fields $\theta$, while parlay trades
train both $\theta$ and $W$ using the trade-implied target price. 

\subsection{Fee-Bias Correction}

Observed Kalshi execution prices are ask quotes rather than direct probability assessments,
and therefore include a platform markup. Writing
\[
p_{\mathrm{obs}}(A) = P^\ast(A) + \varepsilon(A), \qquad \varepsilon(A)\ge 0,
\]
training directly on $p_{\mathrm{obs}}$ would bias the learned marginals and couplings upward.
We therefore preprocess each observed trade into a canonical all-YES target using the
complement rule and inclusion--exclusion. For mixed YES/NO parlays, the additive fee cancels
in the alternating expansion, yielding an unbiased target in the model's native
parameterization. Empirically, this preprocessing is essential: without it, parlay loss
remains flat and revenue stays lower. 

\subsection{Replay Pipeline and Evaluation}

The emulation replays 3,710 candle updates and 5,436 valid executed parlay trades in
chronological order. In the full-market emulation, the observed Kalshi execution price
$p_{\mathrm{exec}}$ is treated as the trader's target price, and the corresponding LMSR
share bundle $(\Delta s_{\text{yes}}, \Delta s_{\text{no}})$ is computed to move the
current quote to $p_{\mathrm{exec}}$. The pool records revenue
$C(s+\Delta s)-C(s)$, and the same event is then used to update the Ising model and execute
shadow trades on the LMSR state. 

At each executed parlay trade, we also compute the pre-trade independence-based quote
$\hat{\theta}$ and the correlation-aware quote $\theta^{\mathrm{corr}}$, and compare them to
the realized Kalshi execution price $p_{\mathrm{exec}}$. We adopt a conservative trader
model in which the trader's reservation price is exactly $p_{\mathrm{exec}}$, so a quote is
accepted only if $\theta^{\mathrm{corr}} \le p_{\mathrm{exec}}$; any quote above the observed
execution price is treated as rejected. 

Table~\ref{tab:pool_sim} evaluates three strategies: Kalshi, Independence,
Stand-alone Market Maker, and Full AMM. For the Stand-alone Market Maker,
pool capital is defined by the maximum drawdown of cumulative premium minus payout; for
Full AMM, it is the sum of per-market LMSR losses. ROI is net PnL divided by required
pool capital, and Sharpe is computed from 5 1-hour event-time portfolio returns.

\end{document}